\documentclass[aps,
    pra,
    superscriptaddress,
    graphicx,
    amsmath,
    amssymb,
    reprint,
    floatfix,
    ]{revtex4-1}
\usepackage[utf8]{inputenc}
\usepackage{graphicx} 
\usepackage{physics}
\usepackage{qcircuit}
\usepackage{xcolor}
\usepackage[normalem]{ulem}
\usepackage{verbatim}
\usepackage{dsfont}
\usepackage{cancel}
\usepackage{enumitem}
\usepackage{amsthm}
\usepackage[colorlinks=true,citecolor=blue]{hyperref}
\usepackage[english]{babel}
\usepackage{orcidlink}

\usepackage[percent]{overpic}

\newcommand{\lf}[0]{\left( }
\newcommand{\ri}[0]{ \right)}
\newcommand{\mean}[1]{\left \langle #1 \right \rangle}
\renewcommand{\ketbra}[2]{\left| #1 \right \rangle \left \langle #2 \right |}
\renewcommand{\braket}[1]{\left \langle #1 \right \rangle}
\newtheorem{prop}{Proposition}
\newtheorem{define}{Definition}
\newtheorem{thm}{Theorem}
\newtheorem*{thm*}{Theorem}

\newtheorem{lemma}{Lemma}
\newtheorem*{lemma*}{Lemma}

\newcommand{\tc}{\tilde{C}}

\def \addCQuIC {Center for Quantum Information and Control, University of New Mexico, Albuquerque, 87131, NM, USA}
\def \addPandAUNM {Department of Physics and Astronomy, University of New Mexico, Albuquerque, NM, 87106, USA}
\begin{document}
\setlength{\abovecaptionskip}{5pt plus 2pt minus 2pt}
\setlength{\textfloatsep}{5pt plus 2pt minus 2pt}

\title{Heisenberg-limited metrology in the presence of non-Markovian noise with finite control rates} 
\author{Shravan Shravan\,\orcidlink{0009-0004-6383-5854}}
\email{shravan@unm.edu}
\affiliation{\addCQuIC} \affiliation{\addPandAUNM}
\author{Tyler G. Thurtell\,\orcidlink{0000-0001-8731-6160}}
\email{tthurtell@unm.edu}
\affiliation{\addCQuIC} \affiliation{\addPandAUNM}

\begin{abstract}
    Recently, it has been shown that protocols utilizing infinitely fast controls, such as quantum error correction, can in principle restore Heisenberg-limited frequency estimation in the presence of a broad class of non-Markovian noise models arising from coupling to finite-dimensional environments. However, these controls differ substantially from those used to address Markovian noise, and their underlying physical mechanism remains unclear. In this work, we establish a direct connection between these protocols and the quantum Zeno effect, and extend the framework to infinite-dimensional environments. We delineate three types of constructions: (a) protocols that rely solely on measurements, (b) protocols using an active recovery after measurements and (c) protocols using dynamical decoupling, and rigorously analyze the performance of each when controls can be applied at only a finite rate. While protocols relying purely on measurements can be engineered for noise models where active recoveries are fundamentally impossible, they result in the quantum Fisher information exhibiting a quadratically worse dependence on the control frequency. Surprisingly, while dynamical decoupling protocols are possible whenever protocols relying only on measurements can be engineered, the quantum Fisher information has the same dependence on control frequency as the active recovery protocols. Numerical simulations suggest that the improvement offered by dynamical decoupling may work in regimes beyond the perturbative setting where our rigorous theorems apply.

\end{abstract}
\maketitle

\section{Introduction}
Quantum metrology is concerned with making high-precision measurements of physical parameters using quantum probes. A typical protocol consists of preparing a probe state that is sensitive to the parameters of interest, allowing it to evolve under the corresponding dynamics, and finally performing a measurement to extract information~\cite{giovannetti_advances_2011}. More generally, the evolution may be interspersed with control operations designed to enhance the estimation precision~\cite{sekatski_quantum_2017}. An essential metrological task is frequency estimation, where the goal is to estimate the frequency $\omega$ associated with a Hamiltonian $H=\omega G$, with $G$ a Hermitian operator. The performance of an estimation scheme is usually quantified by the variance of the associated estimator $\hat{\omega}$, which is often assumed to be unbiased. The quantum Cram\'er-Rao bound sets a fundamental limit on the achievable precision, stating that the variance of any unbiased estimator satisfies
\begin{equation}
    \lf \Delta \hat{\omega} \ri^2 \geq \frac{1}{R\mathcal{F}}.
\end{equation}
where $\mathcal{F}$ is the quantum Fisher information (QFI) and $R$ is the number of independent repetitions of the experiment~\cite{BraunsteinCaves}.

The scaling of the QFI with respect to the total interrogation time $t$ serves as a key benchmark in quantum metrology.  Two fundamental limits are typically distinguished in the literature. If $\mathcal{F}$ asymptotically scales as $t^2$, the precision is said to achieve Heisenberg-limited scaling (HL scaling)~\cite{Higgins_2007}. This is the optimal precision achievable for quantum probes. In the absence of noise, HL scaling can be achieved by preparing a probe state that is a coherent superposition of eigenstates of $G$. In contrast, when the evolution is noisy, decoherence typically leads to an asymptotic decay of the QFI.  In such settings, the typical strategy often consists of restricting the interrogation time to short intervals $\tau \ll \tau_D$, where $\tau_D$ is the characteristic decoherence timescale, and repeating the experiment $R$ times so that $R\tau = t$. For such schemes, $\mathcal{F}$ scales as $\tau^2$ for each repetition, implying that $R\mathcal{F}$ scales as $t$, which is referred to as the standard quantum limit (SQL). This scaling corresponds to the best precision achievable without exploiting quantum coherence and is often equated with the classical limit of estimation. A central question in quantum metrology is whether one can design protocols that mitigate the effects of noise and restore HL scaling in noisy frequency estimation. 

One prominent approach to restoring HL scaling is to use quantum error (QEC) and related control techniques. In standard quantum computation, quantum error-correcting codes (QECCs) are designed so that the local error operators act trivially after projection onto the codespace, while logical operators usually correspond to highly non-local entangling operations~\cite{KnillLaflamme}. By contrast, in quantum metrology, the signal, which may be a local operator, must be encoded  as a non-trivial logical operator. In other words, the codespace of a code constructed for metrology must admit a continuous group of transversal logical operations corresponding to the action of the unitary group generated by the signal acting on the physical qubits.  As a consequence, the Eastin–Knill theorem~\cite{eastin_restrictions_2009} imposes strong restrictions on the noise models for which HL scaling can be recovered using QECCs. Nevertheless, it was previously shown that for specific cases, HL scaling can be achieved using QEC~\cite{sekatski_dynamical_2016,kwon_restoring_2025,tan_enhancement_2013,QEC_for_met_arrad,QEC_for_met_dur,QEC_for_met_kessler,QEC_for_met_ozeri}. The push towards finding ultimate precision limits came through the powerful tool of channel extension bounds~\cite{fujiwara_fibre_2008,sekatski_quantum_2017,demkowicz-dobrzanski_elusive_2012}. In particular, Zhou \textit{et al.}~\cite{zhou_zhang_etal} and Demkowicz-Dobrzanski \textit{et al.}~\cite{demkowicz-dobrzanski_adaptive_2017} established general conditions that Markovian noise must satisfy for HL scaling to be achievable. They show that Heisenberg-limited scaling can be achieved if and only if the signal $G$ and the Lindblad operators of the noise $L_i$ satisfy the Hamiltonian-not-in-Lindblad-span (HNLS) condition, 
\begin{equation}
    G \notin \text{span} \lf L_i, L_i^\dagger, L^\dagger_i L_j  \ri.
\end{equation}
These results were subsequently extended to uncorrelated noise~\cite{zhou_asymptotic_2021} and causal superpositions~\cite{kurdzialek_using_2023}, and were further connected to the notion of transversality and to the approximate Eastin-Knill theorem~\cite{Zhou2021newperspectives, kubica_using_2021}. As a result, precision limits under Markovian noise are now well understood. By contrast, for general non-Markovian noise models, the conditions required to attain HL scaling remain far less clear.

In the literature, non-Markovianity has been defined and quantified in numerous, often inequivalent ways, leading to a lack of a unified framework for studying generic non-Markovian noise~\cite{NM_review,NM_review_2,NM_meas1,NM_meas2}. Nevertheless, non-Markovian metrology has been studied in specific contexts. In Ref.~\cite{ChinHuelgaPlenio} (also studied in Ref.~\cite{macieszczak2015zeno}), Chin \textit{et al.} studied a non-Markovian parallel dephasing model and demonstrated that, by optimizing the interrogation time, the QFI exhibits a  $N^{3/2}$ scaling, where $N$ is the number of probes. This non-standard scaling was attributed to the short time Zeno regime, which was then generalized for \textit{phase-covariant} noise by Smirne \textit{et al.}~\cite{smirne_ultimate_2016}. Altherr and Yang formulated the problem in the framework of quantum combs~\cite{CombsAltherrYoung}. This framework, in principle, provides a fully general description of the non-Markovian setting, however, the computation of such combs becomes intractable for large numbers of controls. Subsequent works~\cite{Efficient_combs_1,Efficient_comb_2} partially address this limitation by showing that, for certain scenarios, the combs can be efficiently simulated using tensor network based methods. More recently, in Ref.~\cite{mann_zhou_laflamme}, Mann \textit{et al.} study the problem using a hidden Markov model, in which a finite-dimensional system and a finite-dimensional environment undergo a joint Lindbladian evolution. In addition to the QEC-based approaches, dynamical decoupling techniques and ideas related to the Zeno effect have also been investigated as a means of mitigating non-Markovian noise in metrology~\cite{sekatski_dynamical_2016,tan_enhancement_2013,lahcen_restoring_2025,zeng2026}. The use of unitary controls to mitigate the effect of specific types of non-Markovian noise in quantum sensing has been studied more recently by Riberi \textit{et al.}~\cite{Riberi_2022,riberi2023,riberi2024,riberi2026}.

In this work, we investigate precision limits for non-Markovian noise models described by a Stinespring dilation. We consider a finite-dimensional system coupled to a generic environment and assume access to the Zeno regime~\cite{misra_sudarshan}. In Theorem \ref{Thm:AQED_zeno}, we first derive sufficient conditions to retrieve HL scaling, in the limit of infinitely fast Zeno controls. For finite-dimensional environments, this reduces to a similar result derived in Ref.~\cite{mann_zhou_laflamme}. We then clarify the role of the recovery operator in our scheme and differentiate between three classes of protocols: (a) ones employing an active recovery, which we refer to as approximate quantum error correction (AQEC), (b) ones without an active recovery, which we refer to as approximate quantum error detection (AQED), and (c) measurement-free dynamical decoupling (DD). These protocols are summarized in Fig~\ref{fig:placeholder}. Pursuant to this, in Theorem \ref{Thm:Recovery_cond} we derive the conditions necessary for constructing an active recovery operator and show that the conditions required in Theorem \ref{Thm:AQED_zeno} are strictly weaker. In Sec.~\ref{sec:finite_meas_freq}, we then rigorously compare the performance of these schemes at finite control frequencies and derive convergence rates for the QFI. In particular, we show that for a time interval of $\delta t \to 0$ between control pulses, the QFI converges as $O(\delta t^2)$ for AQEC and DD protocols, while  only  $O(\delta t)$ convergence is achieved for AQED. Finally, in Sec.~\ref{sec:Num_Sim}, we present numerical simulations that illustrate and support our analytical results.
\begin{figure*}[ht!]
    \centering
    \includegraphics[width=0.48\linewidth]{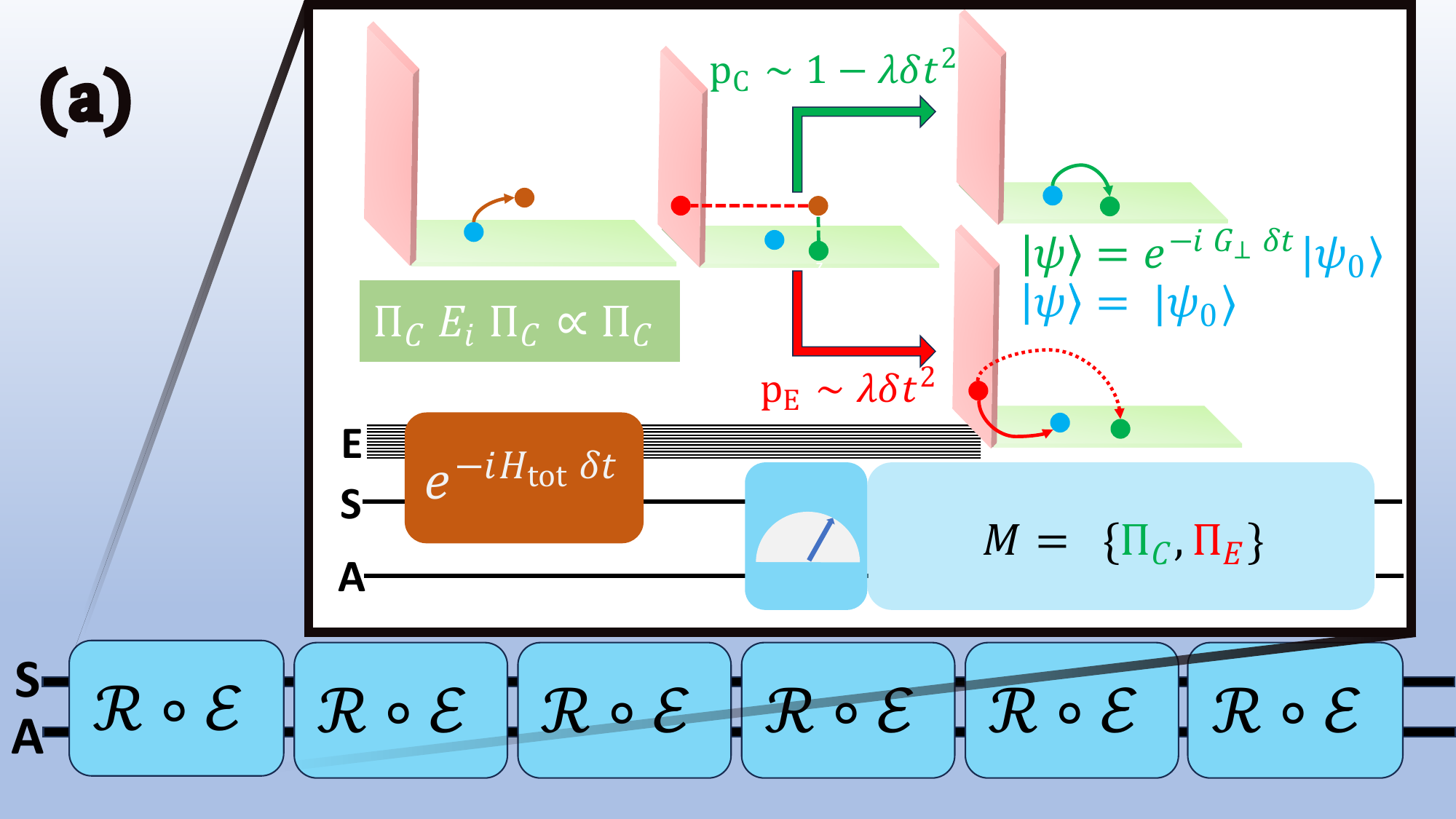}
    \includegraphics[width=0.48\linewidth]{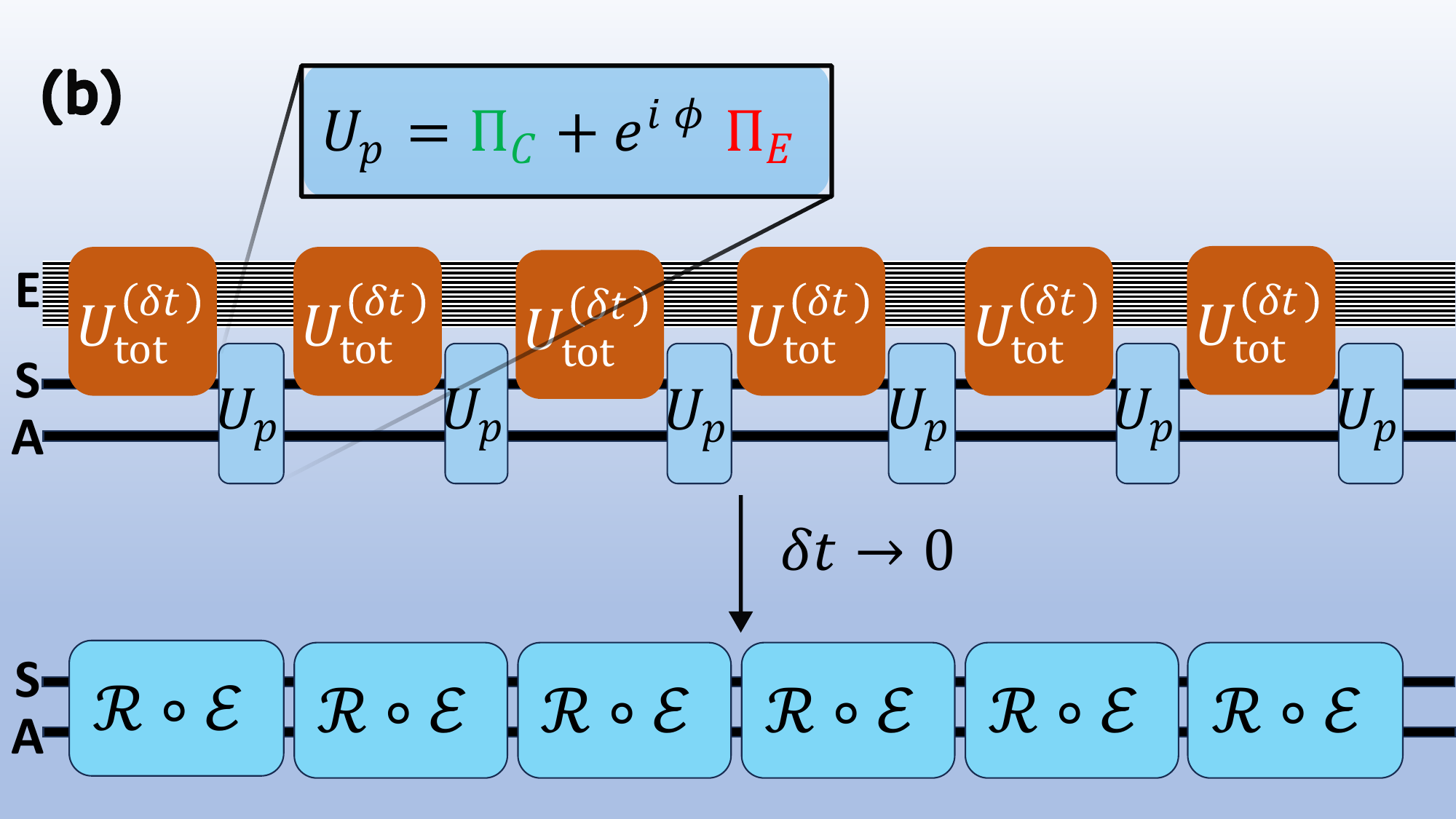}
    \caption{Illustration of the control schemes considered in this work. All lead to HL scaling when the signal generator $G$ is not in the linear span of the error generators $E_{i}$, in the limit of infintely fast pulses. The dynamics of the system and its environment are associated with the Hamiltonian $H_{\textrm{tot}}$ which generates the unitary $U_{\textrm{tot}}(\delta t)$ and leads to the channel $\mathcal{E}$. The second type of dynamics is associated with the controls and leads to the channel $ \mathcal{R}$. (a) Approximate quantum error detection (AQED) and Approximate quantum error correction (AQEC) schemes: After an infinitesimal evolution, a two-outcome projective measurement $M=\{\Pi_{C},\Pi_{E}\}$ is performed. With probability $p_C \sim 1 - \lambda \delta t^2$, the state is projected back into the codespace and undergoes an effective unitary generated by $G_{\perp}$. Otherwise, with probability $p_{E}\sim\lambda\delta t^{2}$, it is projected out of the codespace. In the AQED scheme, we then reset to the initial state $|\psi_{0}\rangle$ (denoted by the solid red arrow). In contrast to the reset based recovery used in the AQED setting, approximate quantum error correction (AQEC) schemes admit a recovery operation which restores the state in the codespace (denoted by dashed red arrow). (b) Dynamical decoupling (DD) scheme: suitably designed control unitaries $U_{P}$ are applied.}
    \label{fig:placeholder}
\end{figure*}

\section{Setup}
\label{sec:setup}
We consider a finite-dimensional quantum system described by the Hilbert space $\mathcal{H}_S$ of dimension $d_S < \infty$, coupled to an environment of arbitrary (possibly infinite) dimension defined on a separable Hilbert space $\mathcal{H}_E$. The joint dynamics are governed by the total Hamiltonian,   
\begin{equation}
     \label{Eq:Setup_Hamiltonian}
     H_{\text{tot}} = \omega_0 G \otimes \mathds{1} + \sum_{k = 1}^{K} g_{k}E_k \otimes B_k + g_{0}\mathds{1} \otimes B_0 + \alpha \mathds{1} \otimes \mathds{1},
\end{equation}
where $G,\{E_k\} \in \mathcal{B}(\mathcal{H}_S)$ are Hermitian operators acting on $\mathcal{H}_S$, and $B_k, B_0$ are (possibly unbounded) self-adjoint  operators on $\mathcal{H}_E$. In particular, the operators $E_k$ are assumed to have norm $\|E_k\| = 1$ and $g_k \in \mathbb{R} $.  The Hamiltonian $H_{\text{tot}}$ is assumed to be self-adjoint in the composite Hilbert space. In the following, we also allow access to an auxiliary $d_S$-dimensional system, henceforth referred to as the noiseless ancilla, which remains decoupled from the environment. Note that the Hamiltonian structure in Eq.~\eqref{Eq:Setup_Hamiltonian} is generic and serves as the starting point for many open quantum system models. Typically, further approximations, based on assumptions regarding timescales of relaxation of the environment and system evolution, lead to effective descriptions such as Lindbladian dynamics in the Markovian regime~\cite{breuer_book}.

Our goal is to estimate $\omega_0$ with a Heisenberg-limited scaling of the QFI. Since such a  scaling can be achieved for a noiseless unitary evolution, our strategy is to suppress the  terms that couple the system to the environment using the quantum Zeno effect. Mathematically, the Zeno effect can be defined as the limiting case of infinitely fast projective measurements. Let $P$ be a single projector corresponding to some measurement $M$, and $U(t) = \exp(-itH)$ denote the unitary evolution generated by a bounded Hamiltonian $H$ on a separable Hilbert space. Dynamics in the Zeno limit are described by
\begin{equation}
   \label{Eq:Zeno_defn}
    \lim_{\delta t\to 0} \left[ P U\lf \delta t\ri P\right]^{\frac{t}{\delta t}} = P\exp(-i t H_Z) P, 
\end{equation}
where $H_z \equiv P H P$ is called the Zeno Hamiltonian~\cite{misra_sudarshan,facchi_pascazio}. For unbounded $H$, the Zeno Hamiltonian must be defined more carefully to ensure that it is self-adjoint~\cite{exner_ichios}. If the initial state is in the subspace defined by the projector $P$, the probability of projecting out of the subspace after $t/\delta t$ measurements of $M$ vanishes as $\delta t \to 0$. To get a more intuitive understanding of the Zeno effect, consider the Taylor expansion of the survival probability of a state $\ket{\psi_0}$ under an infinitesimal evolution $\delta t$,
\begin{equation}
    \left| \braket{\psi_0| U(\delta t) | \psi_0} \right|^2 = 1 - {\text{Var}\lf H \ri } \delta t^2 + O(\delta t^3),
    \label{eq:Practical_Zeno_limit}
\end{equation}
where  $\text{Var}\lf H \ri = \braket{\psi_0| H^2 | \psi_0} - \braket{\psi_0| H | \psi_0}$ and $U$ is the unitary generated by $H$. Thus, in the Zeno limit, the state evolution is of order $\delta t^2$. Hence, if we perform a projective measurement given by projectors $\{\ketbra{\psi_0}{\psi_0}, \mathds{1} - \ketbra{\psi_0}{\psi_0}\}$, at a time of order $\delta t$ we freeze the evolution of the system~\cite{misra_sudarshan,exner_ichios}. The use of Taylor approximation here is justified by the assumption that $\Omega\delta t$ is small where $\Omega$ is a stand in for any timescale associated with the Hamiltonian $H$. In other words, this limit strictly exists only when the Zeno measurement timescale is the fastest. In particular, we are not in a regime where Markovian approximations are valid. We remark that, though this seems to be a very stringent condition, it is regularly assumed in the literature of dynamical decoupling pulses and related control protocols. 

\section{Zeno limit scheme for retrieving HL scaling}
\label{sec:limits}
In this section, we state our first main result regarding achieving a Heisenberg-limited scaling in the estimation of the parameter $\omega_0$ using controls and measurements only on the system and ancilla. Before stating the theorem, we will first state a Lemma, using which we will construct error correcting codes for metrology. This construction is given in Refs.~\cite{zhou_zhang_etal,mann_zhou_laflamme}. However, since this construction is essential in the proof, for the sake of completeness, we describe it in Lemma~\ref{lem:Error_detection}. Let $G$ be a Hermitian operator, and $\{E_k| k = 1,\dots, K\}$ be an orthonormal set of traceless Hermitian operators such that 
\begin{equation}
\label{eq:hnls}
    G \notin \text{span}_{\mathbb{H}}\{E_k, \mathds{1}\}, 
\end{equation}
where $\text{span}_{\mathbb{H}}(S)$ refers to the set of all Hermitian operators that are spanned by elements of the set $S$. Define the operator $G_\perp$ as 
    \begin{equation}
    \label{Eq:G_perp}
        G_\perp = G - \sum_k \operatorname{Tr} \lf G E_k\ri E_k - \operatorname{Tr}(G) \frac{\mathds{1}}{d_S}.
    \end{equation}
 Since $G \notin \text{span}_{\mathbb{H}} \{E_k, \mathds{1}\}$, we see that $G_\perp$ is a traceless nontrivial Hermitian operator. Hence, using the spectral theorem, we can write 
\begin{equation}
    \label{Eq:G_perp_spec_decomp}
    G_\perp = \rho_+ - \rho_-, 
\end{equation}
for positive operators $\rho_+, \rho_-$ with equal trace.
\begin{lemma}[\cite{zhou_zhang_etal,mann_zhou_laflamme}]
\label{lem:Error_detection}
  Let $G$ satisfy Eq.~\eqref{eq:hnls}, $G_\perp$ be defined as in Eq.~\eqref{Eq:G_perp}, and $\rho_{+}$ and $\rho_{-}$ be positive operators that satisfy Eq.~\eqref{Eq:G_perp_spec_decomp}. Define the states $\ket{\psi_+}, \ket{\psi_-}  \in \mathcal{H}_S \otimes \mathcal{H}_S$ to be purifications of $\rho_+$ and $\rho_-$. Now, define the projection operator 
    \begin{equation}
        \Pi_C = \ketbra{\psi_+}{\psi_+} + \ketbra{\psi_-}{\psi_-}.
    \end{equation}
    The projection operator $\Pi_C$  satisfies 
    \begin{equation}
        \Pi_C G \Pi_C \not\propto \Pi_C, \quad \Pi_C E_k \Pi_C = \lambda_k \Pi_C ~ \forall k, \lambda_k \in \mathds{R}. 
    \end{equation}
\end{lemma}
We provide the proof of the above lemma in App.~\ref{ap:Proof_theorem_1}. With the lemma stated, we state our first main result. 
\begin{thm}
\label{Thm:AQED_zeno}
Let $\mathcal{{H}}_S$ be a finite-dimensional Hilbert space for the system and  $\mathcal{H}_E$ a separable Hilbert space for the environment. Consider the joint evolution on $\mathcal{{H}}_S \otimes \mathcal{{H}}_E$ generated by the Hamiltonian in Eq.~\eqref{Eq:Setup_Hamiltonian}.  
 Assume:
 \begin{enumerate}
     \item  The Hamiltonian $H_{\text{tot}}$ is non-negative,
     \item  One has access to noiseless ancillas and the ability to perform arbitrary unitary controls and measurements on the system and ancilla.
 \end{enumerate}
Under these assumptions, we can achieve a Heisenberg-limited scaling of the quantum Fisher information of the parameter $\omega_0$, i.e.,
 \begin{equation}
    \lim_{\delta t \to 0} \mathcal{F}(t) = \Theta(t^2) ,
 \end{equation}
provided that the signal generator $G$ satisfies the not-in-span condition 
   \begin{equation}
       G \notin \text{span}_{\mathbb{H}}\lf E_k , \mathds{1}\ri ~ \mathrm{for~all} ~~ k \in \{1, \dots, K\},
       \label{eq:span_cond}
   \end{equation}
   and that the symmetric operator $P_0 H_{tot} P_0$ is essentially self adjoint. Here $P_0 = \Pi_C \otimes 1_E$ is a projector, where $\Pi_C$ is constructed as per Lemma \ref{lem:Error_detection}.
\end{thm}

The not-in-span condition in Eq.~\eqref{eq:span_cond} generally holds for any operator decomposition of the coupling Hamiltonian, but it is strongest when $B_k$ are linearly independent. The non-negativity condition on $H_{tot}$ can be weakened: it suffices that its eigenvalues are bounded from below. Mathematically, this is necessary since the Zeno limit is ill-defined for Hamiltonians whose spectrum is unbounded below. Physically, we expect all Hamiltonians to satisfy this condition, as otherwise various unphysical phenomena can occur. For example, the standard partition function would diverge. 

The proof of this theorem is quite involved and is presented in the App.~\ref{ap:Proof_theorem_1}. However, we give a high-level summary here. First, we construct a codespace for which the error generators $E_k$ act trivially on the logical degrees of freedom while the signal $G$ acts non-trivially on these degrees of freedom, using Lemma \ref{lem:Error_detection}. Note that in this codespace, the effective evolution is generated by a decoupled Hamiltonian of the form
    \begin{equation}
            H_{\text{eff}} = \omega G_{\textrm{eff}} \otimes \mathds{1} + \mathds{1} \otimes \tilde{B},
    \end{equation}   
where $G_{\textrm{eff}} = \Pi_C \mathds{1}_A \otimes  G_{\perp} \Pi_C$. Thus, as long as we remain in this codespace, the evolution is effectively decoupled. We achieve this by constantly measuring and projecting the state into the codespace, utilizing the arguments of the Zeno effect \cite{misra_sudarshan,facchi_pascazio}. This ensures the effective evolution remains unitary.  

The above theorem reduces to the special case of unitary evolution in Ref.~\cite{mann_zhou_laflamme} if the environment is finite-dimensional. In Theorem 4 of that reference, the authors show that for such systems the QFI scales as $t^2$ times a periodic envelope, even if the span conditions are not met. This happens because finite-dimensional environments exhibit Poincar\'e recurrences. There is a sequence of times $\{t_n\}$ where the evolution appears purely unitary. This occurs for any Hamiltonian with only a discrete point spectrum \cite{Rec_Thm_1,Rec_Thm_2}. By contrast, infinite-dimensional environments with a continuous spectrum lack recurrences, so the span conditions discussed above may be necessary, in addition to sufficient, to achieve $t^2$ scaling.  We conjecture that the span conditions are also necessary for such cases, but leave this as an open problem.

Finally, we highlight a technical subtlety. Theorem \ref{Thm:AQED_zeno} is formally applicable only when the joint Hilbert space is separable, i.e, it admits a countable basis. In the standard open quantum formalism, however, a continuous mode limit is adopted, wherein the system is coupled to a continuum of bosonic fields \cite{breuer_book}. The Hilbert space for such a space need not be separable. Nevertheless, the above theorem is still suggestive of the physics in continuous mode settings if we interpret the continuous mode limit as the thermodynamic limit (or limit of infinite volume) of a finite number of bosonic fields, for which the Hilbert space is separable.

\subsection{Error detection vs. correction}

At this juncture, we wish to emphasize the difference between the conventional error correction schemes for sensing (such as those employed to mitigate Markovian noise in Ref.~\cite{zhou_zhang_etal}) and the scheme proposed in Theorem~\ref{Thm:AQED_zeno}. In settings where conventional error correction schemes are applicable, the infinitesimal noise model is typically taken to be of the form 
\begin{equation}
    \mathcal{E}(\rho) = (1 - p) \rho + p\tilde{\mathcal{E}}(\rho), 
\end{equation}
where the channel $\tilde{\mathcal{E}}$ has Kraus operators $\tilde{K}_i$ for $i \in \{1, \dots, n\}$. The channel $\mathcal{E}$ can be interpreted as applying the noisy channel $\tilde{\mathcal{E}}$ with a probability $p$ with the Kraus operators $K_0 = \sqrt{1 - p}\mathds{1}$ and  $K_i = \sqrt{p}\tilde{K}_i$ for $i \in \{1,\dots,n\}$. For example, in the case of unital noise on single qubit described by a Lindbladian master equation, one has $p = \gamma \delta t$, where, as before, $\delta t$ is the time interval between two successive control operations and $\gamma$ is the effective decay rate and the Kraus operators $\tilde{K}_i = L_i + O(\delta t)$ are the Lindblad operators. In this framework, one constructs a codespace that satisfies the Knill-Laflamme conditions for the channel $\mathcal{E}$, resulting in the conditions 
\begin{align}
\label{Eq:KL_cond_gen_channel_a}
  \Pi_C \tilde{K}_i\Pi_C & = \lambda_i \Pi_C, \\
  \label{Eq:KL_cond_gen_channel_b}
  \Pi_C \tilde{K}_i^\dagger \tilde{K}_j \Pi_C &= \alpha_{ij} \Pi_C.
\end{align}
For any state $\rho$ in the codespace, $\Pi_C \rho \Pi_C = \rho$ and the conditions in Eqs.~\eqref{Eq:KL_cond_gen_channel_a} and \eqref{Eq:KL_cond_gen_channel_b} imply
\begin{equation}
    \label{Eq:Channel_in_codespace}
    \Pi_C \mathcal{E}(\rho) \Pi_C = (1 - p \Lambda) \rho,
\end{equation}
where $\Lambda = 1 - \sum_i |\lambda_i|^2$ and $(1 - p\Lambda) \leq 1$. Moreover, the effective channel after projection outside the codespace, 
\begin{equation}
    \tilde{\mathcal{C}}(\sigma) := \Pi_{\tilde{C}} \mathcal{E}(\Pi_C \sigma \Pi_C)\Pi_{\tilde{C}},
\end{equation}
where $\Pi_{{\tilde{C}}}=\mathds{1}-\Pi_{C}$, also satisfies the Knill-Laflamme conditions, which implies the existence of a recovery operator $\mathcal{R}_{\tilde{\mathcal{C}}}$ such that $\mathcal{R}_{\tilde{\mathcal{C}}} \circ \tilde{\mathcal{C}} (\rho) \propto \rho$. Thus, the recovery operator for the channel $\mathcal{E}$ is given as, 
\begin{equation}
   \label{Eq:Recovery_form}
    \mathcal{R}(\sigma) = \Pi_C \sigma \Pi_C + \mathcal{R}_{\tilde{\mathcal{C}}} (\Pi_{\tilde{C}} \sigma \Pi_{\tilde{C}}).
\end{equation}
The structure of Eq.~\eqref{Eq:Recovery_form} highlights two features: (a) when the measurement projects the state to the codespace the pre-error state $\sigma$ is obtained and (b) when the state is projected outside the codespace, there exists a recovery operation that restores it. Thus, the conventional error correction scheme proceeds in two steps: (a) detect whether the state has left the codespace, and (b) apply the correction if necessary. Heisenberg-limited scaling is possible in this setting, under the assumption of full and fast controls, if the signal generator acts non-trivially in the codespace \cite{Zhang_Liang_HNKS}. 

By contrast, Theorem~\ref{Thm:AQED_zeno} utilizes the Zeno effect and in the limiting case of $\delta t \to 0$ continuously projects the system back into the codespace. This ensures that the evolution remains noiseless at all times. Thus, in this case, the recovery operator $\mathcal{R}_{\tilde{\mathcal{C}}}$ is unnecessary: frequent error detection alone suffices, provided it is performed rapidly enough. Formally, the Zeno scheme implements an error detection scheme as opposed to an error correction scheme. This can be seen by the fact that the codespace defined in Lemma \ref{lem:Error_detection} is guaranteed to satisfy Eq.~\eqref{Eq:KL_cond_gen_channel_a} but may or may not satisfy Eq.~\eqref{Eq:KL_cond_gen_channel_b}.   This distinction has two notable consequences. First, the conditions required for achieving Heisenberg-limited scaling under the Zeno scheme are strictly weaker than those of the standard error correction framework. Consequently, the Zeno scheme achieves Heisenberg-limited scaling for a strictly larger set of error channels. Second, as we will show, in realistic scenarios where $\delta t$ is finite, the ability to correct errors becomes critical: the Zeno-based error detection scheme is less robust to finite detection intervals compared to its error correction counterpart. 

For a finite $\delta t$, there are two sources of errors in the Zeno scheme that we have to contend with. First, there is an effective parallel dephasing on the codespace which can occur for both AQEC and AQED schemes. This is inherently uncorrectable because the error generator is parallel to the signal generator. The second source of error is due to the state being projected out of the codespace upon performing the projective measurement $\{\Pi_C, \Pi_{\tilde{C}}\}$. Unlike the parallel dephasing, this type of error can be corrected if there exists a recovery operator $\mathcal{R}_{\tilde{\mathcal{C}}}$ of the type that appears in Eq.~\eqref{Eq:Recovery_form}. In Theorem \ref{Thm:Recovery_cond} we show that the not-in-span condition in Theorem \ref{Thm:AQED_zeno} is not by itself sufficient to ensure that such a recovery exists. We further identify the additional constraints that the error generators must satisfy to ensure recoverability.

\begin{thm}
\label{Thm:Recovery_cond}
    Consider the setup and joint Hamiltonian as in Eq.~\eqref{Eq:Setup_Hamiltonian}. Define the channel $\mathcal{E}$ as
    \begin{equation}
        \mathcal{E}_{\delta t}(\rho) := \operatorname{Tr}_E \lf \mathds{1}_A \otimes U_{\text{tot}}(\delta t) \rho \otimes \rho_B \mathds{1}_A \otimes U^\dagger_{\text{tot}}(\delta t)\ri,
    \end{equation}
    where $U_{\text{tot}}(\delta t) = \exp(-i \delta t H_{\text{tot}})$ is the unitary generated by $H_{\text{tot}}$. Suppose there exists a positive constant $C >0$, such that the $n$-point bath correlation functions satisfy 
    \begin{equation}
        \left|\operatorname{Tr}_E\lf B_{i_1} \dots B_{i_n} \rho_B \ri \right| \leq C^n \quad \forall n \in \mathbb{N}.
    \end{equation}
    Then the code constructed as per Lemma \ref{lem:Error_detection} for recovering Heisenberg-limited scaling admits a recovery operation $\mathcal{R}_{\tilde{\mathcal{C}}}$ such that
    \begin{equation}
        \mathcal{R}_\mathcal{\tilde{C}} \lf \Pi_{\tilde{C}} \mathcal{E}_{\delta t}(\rho) \Pi_{\tilde{C}} \ri \propto \rho + O(\delta t^3), \qquad \text{for} ~ \rho = \Pi_C \rho \Pi_C, 
    \end{equation}
    only if 
    \begin{align}
         \label{eq:Addl_cond_1}
         \Pi_C &\mathds{1}_A \otimes E_j E_k \Pi_C \propto \Pi_C, \\
         \label{eq:Addl_cond_2}
         \Pi_C &\mathds{1}_A \otimes E_{j}  \Pi_{\tilde{C}}  \mathds{1}_A \otimes G_\perp \Pi_C \propto \Pi_C, \\
         \label{eq:Addl_cond_3}
         \Pi_C &\mathds{1}_A \otimes G_\perp  \Pi_{\tilde{C}}  \mathds{1}_A \otimes G_\perp \Pi_C \propto \Pi_C,
    \end{align}
    with the first two proportionalities holding for all $j$ and $k$.
\end{thm}

The proof of this theorem is given in  App.~\ref{ap:Proof_theorem_2}. Theorem~\ref{Thm:Recovery_cond} tells us that an active recovery operator $\mathcal{R}_{\tilde{\mathcal{C}}}$ codespace can be constructed only if $\Pi_{C}E_{j}E_{k}\Pi_{C}\propto\Pi_{C}$ while we get non-trivial dynamics only if $\Pi_{C}G\Pi_{C}\not\propto\Pi_{c}$. Therefore, we require
\begin{equation}
    \label{Eq:quadratic_span}
    G \notin \text{span} \lf \mathds{1}, E_i, E_j^\dagger, E_i^\dagger E_j \ri.
\end{equation}
The condition in Eq.~\eqref{Eq:quadratic_span} is similar in essence to the HNLS condition derived in Ref.~\cite{zhou_zhang_etal} for Markovian noise. Crucially, this is a stronger condition than the one demanded in Eq.~\eqref{eq:span_cond} and hence there are noise models wherein a Heisenberg-limited scaling can be achieved given access to Zeno-limited controls but for which it is impossible to construct an active recovery operator. A paradigmatic example is sensing a rotation along $\sigma_z$ under amplitude damping noise. The noise is generated by operators that are coupled to the system through $\sigma_x$ and $\sigma_y$, thus satisfying the condition in  Eq.~\eqref{eq:span_cond} but not the condition in Eq.~\eqref{Eq:quadratic_span}. We stress the codes for models that only satisfy Eq.~\eqref{eq:span_cond} and the ones that meet the additional condition Eq.~\eqref{Eq:quadratic_span} are constructed in the same manner, as in Lemma~\ref{lem:Error_detection}. The difference is that while for the former case the code is only error detecting, for the latter it becomes error correcting. We distinguish the two by calling the codes approximate quantum error detecting (AQED) if they do not satisfy Eq.~\eqref{Eq:quadratic_span}, and approximate quantum error correcting (AQEC) if they do satisfy Eq.~\eqref{Eq:quadratic_span}. We note that it is sometimes possible to construct codes of interest for error corrected metrology in a manner different from that described in Lemma~\ref{lem:Error_detection}. For example, stabilizer codes~\cite{Antu_2025}, permutation-invariant codes~\cite{ouyang2026}, and subsystem codes~\cite{liu2026} have been proposed for use in quantum metrology. However, for the sake of concreteness we focus on this family of codes which always suffice to recover Heisenberg scaling when the needed span conditions are satisfied and the control timescales are sufficiently faster than all other timescales in the problem.

\subsection{Unitary pulse scheme}

As shown in Theorem \ref{Thm:AQED_zeno}, the AQED scheme achieves the Zeno limit under infinitely frequent projective measurements. The use of projective measurements, however, is not essential and the same dynamics can instead be engineered using unitary pulse sequences~\cite{Unification_of_DD_and_QZE}. To see this, consider a codespace constructed according to Lemma \ref{lem:Error_detection}. Define the unitary operator 
\begin{equation}
    U_p =  \Pi_C + e^{i \phi} \Pi_{\tilde{C}},
\end{equation}
where, $\phi \neq 0\pmod{2\pi}$. Consider a fast pulsing scheme where the unitary $U_p$ is applied after each interval of free evolution under $H_{\text{tot}}$  for a duration $t/n$. For a fixed $n$, the corresponding total time evolution operator is given by
\begin{equation}
    T_n(t) = \lf U_p U_{\textrm{tot}}\lf\frac{t}{n} \ri \ri^n.
\end{equation}
Then, following the prescription of Ref.~\cite{Unification_of_DD_and_QZE}, one can show that for any bounded $H \in \mathcal{B}{(\mathcal{H}_S \otimes \mathcal{H}_B)}$   
\begin{align}
   T(t) &:=  \lim_{n \to \infty} T_n(t) \nonumber \\ 
   &= \exp \lf -i t \lf  \Pi_C H_{\textrm{tot}} \Pi_C + \Pi_{\tilde{C}} H_{\textrm{tot}} \Pi_{\tilde{C}} \ri \ri.
\end{align}
 This result was later extended to apply to a more general class of Hamiltonians on infinite-dimensional Hilbert spaces by Bernád in Ref.~\cite{Bernad_unitary_to_zeno_1} (Theorem 3). To distinguish this from the measurement based AQED scheme, we call the unitary pulse implementation the dynamical decoupling (DD) scheme.   

The DD scheme, while formally equivalent to the measurement-based AQED scheme in the limit $\delta t \to 0$, offers several advantages over the latter. Since the DD pulse sequence is constructed from the AQED projection operators, the condition in Eq.~\eqref{eq:span_cond} is sufficient to guarantee its existence.  However, as we show in the following section, the DD scheme is, remarkably, more robust to finite pulse frequencies, achieving a scaling equivalent to that of the AQEC scheme. From a practical perspective, the implementation of both AQED and AQEC is challenging, particularly on noisy intermediate-scale quantum (NISQ) devices, owing to the difficulty of performing high-fidelity mid-circuit measurements~\cite{Negnevitsky_2018,Ryan_Anderson2021,Egan2021,Deist2022,Singh_2023,Graham2023,Lis2023,Norcia2023,Huie2023,Corcoles_2021,Koh_2023,Acharya2024}.  By contrast, current quantum platforms typically implement unitary operations with significantly higher fidelities and on much shorter timescales than mid-circuit measurements, making the DD scheme particularly well suited to the NISQ era.

\section{Finite measurement frequency}
\label{sec:finite_meas_freq}
In the preceding section, we were concerned only with whether the necessary controls existed in the limit where control operations could be performed in arbitrarily rapid succession. Mathematically, this means that we studied the convergence of a sequence of time evolution operators in the limit $\delta t \to0$, where $\delta t$ is defined as the time duration between successive rounds of measurement. However, from a practical standpoint, it is desirable to understand how finite measurement frequency affects our results due to the finite timescales of realistic measurement. Thus, in practical sense, we are not only interested in convergence, but also the rate of convergence of these sequences. Since analysis of rate of convergence becomes complicated for unbounded Hamiltonians, in this section, we will only focus our attention on bounded Hamiltonians.  Note that the introduction of a finite timescale $\delta t$ between control operations is distinct from allowing those control operations to take a finite amount of time. New and distinct challenges can arise in both settings. In this work we restrict our focus to the first challenge.

\subsection{AQEC codes}

Protocols that utilize an AQEC code are more robust to finite measurement frequency than their AQED counterparts due to the existence of a recovery operator $\mathcal{R}_\mathcal{\tilde{C}}$. To gain an intuition, consider the evolution of a state in the codespace such that at time $t + \delta t$, the state collapses into the error space. Since a recovery operation exists, it restores the state to its form at time $t$. Thus, the system has neither gained nor lost any information between $t$ and $t+\delta t$. Contrast this with the AQED, where it is often the case that the only reasonable ``recovery" is to  reset the state back to its initial state (see App.~\ref{Ap:Reset_conds}). Here, we have lost all information that was acquired from times $0$ to $t$. This intuition can be rigorously proven. Under the assumption that the $n$-point bath correlation function of the environment state is bounded by $C^n$ for some positive constant $C$, we rigorously calculate the error in the AQEC scheme in Theorem \ref{Thm:Error_Order_AQEC}.

\begin{thm}
     \label{Thm:Error_Order_AQEC}
   Consider the same setup and joint Hamiltonian as in Eq.~\eqref{Eq:Setup_Hamiltonian}. Let $U_{\text{tot}}(\delta t)$ be the unitary generated by the Hamiltonian, $U_{\textrm{tot}}(\delta t):= \exp(-i \delta t H_\text{tot}) $ and let $\mathcal{U}$ be the superoperator associated with this unitary. Suppose there exists a positive constant $C >0$, such that the $n$-point bath correlation functions satisfy 
    \begin{equation}
        \left|\operatorname{Tr}_E\lf B_{i_1} \dots B_{i_n} \rho_B \ri \right| \leq C^n.
    \end{equation}
    Define the channel $\mathcal{E}$ as
    \begin{equation}
        \mathcal{E}(\rho) := \operatorname{Tr}_E  \lf \mathds{1}_A \otimes U_{\text{tot}}(\delta t) \rho \mathds{1}_A \otimes U^\dagger_{\text{tot}}(\delta t)\ri.
    \end{equation}
  Suppose that the codespace constructed according to Lemma~\ref{lem:Error_detection} satisfies the conditions in Eqs. \eqref{eq:Addl_cond_1},\eqref{eq:Addl_cond_2} and \eqref{eq:Addl_cond_3}, which using Theorem~\ref{Thm:Recovery_cond} implies the existence of a channel $\mathcal{R}$,
    \begin{align}
        \mathcal{R} \circ \mathcal{E} (\rho_S \otimes \rho_B)  &= \Pi_C \mathcal{E} (\rho_S \otimes \rho_B) \Pi_C \\
        &\qquad + \mathcal{R}_{\tilde{\mathcal{C}}} \lf\Pi_{{\tilde{C}}} \mathcal{U} (\rho_S \otimes \rho_B) \Pi_{{\tilde{C}}} \ri,
    \end{align}
   where $\mathcal{R}_{\tilde{\mathcal{C}}}$ satisfies 
   \begin{equation}
         \label{EQ:Errorspace_recovery}
        \mathcal{R}_{\tilde{\mathcal{C}}} \lf \Pi_{{\tilde{C}}} \mathcal{E}(\rho) \Pi_{{\tilde{C}}}\ri \propto \rho + O(\delta t^3), \quad \text{for} ~ \rho = \Pi_C \rho \Pi_C.
    \end{equation}
   Then,  we have 
   \begin{align}
       \mathcal{R} \circ \mathcal{E} (\rho_S \otimes \rho_B)  &=  V_{\text{exact}}(\delta t) \rho_s  V^\dagger_{\text{exact}}(\delta t) +  O(\delta t^3),
   \end{align}
   where the unitary $V_{\text{exact}}(\delta t)$ is given as 
   \begin{equation}
   \label{eq:exact_evo}
       V_{\text{exact}}(\delta t) := \exp\left[ -i \delta t \lf \mathds{1}_A \otimes G_{\mathrm{eff}} \ri \right] \Pi_C.
   \end{equation}
\end{thm}
The proof of the above theorem is given in App.~\ref{Ap:Proof_Theorem_3}. First, note that the theorem assumes that the $n$-point correlation functions are bounded, which is a much weaker assumption than demanding the Hamiltonian be bounded. Hence,  this theorem works in more general settings. Second, this says that the evolution operator is equivalent to a unitary evolution up to the $O(\delta t^3)$. Under the mild assumption that the effective bath state at time $t$, $\rho_B(t)$ has the bounded correlation property, Theorem~\ref{Thm:Error_Order_AQEC} can be repeatedly applied to show that at time $t$, the evolution is unitary to order $O(t \delta t^2)$. This immediately implies that the QFI scales as $\mathcal{F}(t) \approx 4 t^2 \operatorname{Var} \lf\mathds{1}_A \otimes G_{\textrm{eff}} \ri |_{\ket{\psi_0}}$, for at least a time
\begin{equation}
t=\Omega\left(\frac{1}{\nu^{3}\delta t^{2}}\right)
\end{equation}
where $\nu$ is a frequency that depends on the details of the Hamiltonian. Through our numerical simulations, we show that this scaling is indeed tight.

\subsection{AQED codes and DD schemes}
\label{subsec:AQED}
Unlike the noise models that admit an AQEC code, for noise models that admit only an AQED code, a recovery operator $\mathcal{R}_{\tilde{\mathcal{C}}}$ {with the properties outlined above} does not exist. In certain cases, spontaneous decay, for example (see App.~\ref{Ap:Reset_conds}), the state after projection onto the error space has no information about the parameter $\omega_0$. Thus, in the worst-case scenario, {a reasonable} option is to discard the state if it is projected onto the error subspace and reset to the initial state at $t = 0$. In App.~\ref{Ap:Reset_conds} we derive the generic conditions under which the state in the error space has no information about the parameter. Consequently, a projection into the error space in this setting entails the loss of all information accumulated up to that time. Intuitively, this means that the action of projecting onto the error space is worse for models that admit only an AQED code when compared with those admitting an AQEC code. To put this intuition on a more mathematically rigorous ground, we perform the leading order error analysis for noise models that admit an AQED code in the App.~\ref{Ap:Error_AQED}. We show that for bounded Hamiltonian $H_{\text{tot}}$, we have 
\begin{align}
\Pi_C \big( \mathds{1}_A \!\otimes\! U(\delta t) \big) \Pi_C
= &\exp\Big[
   - i\,\delta t\,
     \Pi_C \big( \mathds{1}_A \!\otimes\! H_{\mathrm{tot}} \big) \Pi_C
\notag \\[4pt]
&
   - \frac{\delta t^2}{2}\,
     \Pi_C \big( \mathds{1}_A \!\otimes\! H_{\mathrm{tot}} \big) \Pi_{\tilde{C}}
      \big( \mathds{1}_A \!\otimes\! H_{\mathrm{tot}} \big) \Pi_C
\Big]
\notag \\[4pt]
& +\, O\!\left(\|H_{\mathrm{tot}}\|^3 \delta t^3\right).
\label{Eq:V_AQED_perturb}
\end{align}
Note that in Eq.~\eqref{Eq:V_AQED_perturb} the leading order of error is of the order $O(\delta t^2)$ as opposed to $O(\delta t^3)$ for AQEC codes. Hence, the QFI scaling as $\mathcal{F}_{\text{AQED}}(t) \approx 4 t^2 \operatorname{Var} \lf\mathds{1}_A \otimes G_{\textrm{eff}} \ri |_{\ket{\psi_0}}$ is only guaranteed until a time
\begin{equation}
t=\Omega\left(\frac{1}{\nu^{2}\delta t}\right).
\end{equation}
Furthermore, this scaling is tight as illustrated by the example of perpendicular dephasing, where the infidelity and therefore the leading order error in Fisher information scales as $O(t \delta t)$. In App.~\ref{ap:Leading_order_DD}, we perform a similar analysis for 
the unitary pulse scheme and find that, for $\phi = \pi$,
\begin{equation}
   \label{eq:DD_error}
   \begin{split}
    T_2(t) = \exp&\lf-i \delta t \lf 2 H_0 \ri - \frac{\delta t^2}{2} [H_0, A_0] \ri \\
    &+ O\lf\|H_{\text{tot}}\|^3\delta t^3\ri,
    \end{split}
\end{equation}
where
\begin{align}
H_{0}&=\Pi_{C}H_{\textrm{tot}}\Pi_{C}+\Pi_{\tilde{C}}H_{\textrm{tot}}\Pi_{\tilde{C}}, \\
A_{0}&=\Pi_{\tilde{C}}H_{\textrm{tot}}\Pi_{C}+\Pi_{C}H_{\textrm{tot}}\Pi_{\tilde{C}}.
\end{align}
From the above equation, note that the while the leading order error is of the order $O(\delta t^2)$, it has a starkly different form. While in the AQED case, the leading order error term couples the codespace to itself, in the unitary pulse case it couples the codespace to the error space. Consider a measurement with POVM's having the form $\{\Pi_C^{(i)}, \Pi_{\tilde{C}}\}$, where the $\Pi_{C}^{i}$ are supported only on the codespace and $\Pi_{\tilde{C}}$ is not supported on the codespace. For a noiseless evolution, there must exist a POVM of this form that maximizes the classical Fisher information. In the density matrix description, the lowest order error for evolution under Eq.\eqref{eq:DD_error} is $O(t \delta t)$, but appears only in the coherences between the codespace and the error space, while the error within the codespace is $O(t \delta t^2)$. Hence, under the POVM $\{\Pi_C^{(i)}, \Pi_{\tilde{C}}\}$, the leading order error in the associated probability distribution is $O(t \delta t^2)$. This manifests itself at the level of the QFI and we can show (see App.~\ref{ap:Leading_order_DD}) that for the unitary pulse scheme, the QFI scales as $\mathcal{F}_{\text{DD}}(t) = 4 t^2 \operatorname{Var} \lf\mathds{1}_A \otimes G_{\perp} \ri|_{\ket{\psi_0}}$ until at least times
\begin{equation}
t=\Omega\left(\frac{1}{ \nu^{3}\delta t^2}\right).
\end{equation}

The leading order errors are accurate when $\delta t$ is the fastest timescale in the system-environment model. The performance of DD, AQEC, and AQED codes might worsen if this is not the case.  For example, if $\delta t > \tau_B$, the reservoir scrambling time, then the noise model is effectively Markovian and hence Zeno controls would fail.  Similarly, for $\delta t$ large enough, we might be in the so called anti-Zeno effect, which decreases the probability of projecting onto the codespace instead of increasing them \cite{FacchiTasaki,KofmannKuriziki,QingYongZheng}. Since the effect of projecting out of the codespace is worse for noise models admitting only an AQED code vs those admitting an AQEC code, AQED codes should perform markedly worse in such regimes.

\section{Numerical simulations}
\label{sec:Num_Sim}

\begin{figure}[]
\includegraphics[width=\linewidth]{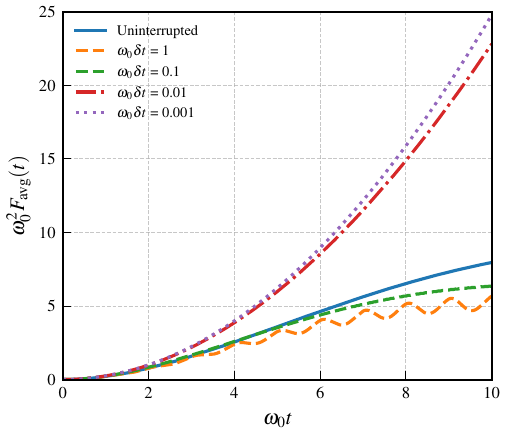}
    \caption{QFI as a function of time for Jaynes-Cummings model with detuning. We have set the parameters $\lambda = \omega_0$, $\gamma_0 = \Delta = 5 \omega_0$. $\delta t$ is the time duration between successive measurements. We see a decreasing value of the QFI with increasing $\delta t$. Note that the QFI for $\omega_0 \delta t = 0.1$ is smaller than for the uninterrupted evolution.}
\label{fig:combined_image}
\end{figure}

To highlight our assertions in this article we perform numerical simulations of two different noise models to reinforce our results. We first simulate the model of non-Markovian spontaneous emission, which admits only an AQED code. Consider the Hamiltonian 
\begin{equation}
\label{eq:jaynes_cummings}
    H = \omega_0 \sigma_+\sigma_- + g_{0}\mathds{1}\otimes B_{0} + (g_{1}\sigma_+ \otimes B_{1}  + \text{h.c.}), 
\end{equation}
where $\sigma_{\pm}$ are the qubit raising and lower operators. In particular the standard Jaynes-Cummings model in the rotating wave approximation, is of this form. 
This model can be understood as describing the evolution of a qubit under $\sigma_z$ and amplitude damping noise, generated by $\sigma_+$ and $\sigma_-$~\cite{breuer_book}. From Eqs.~\eqref{eq:span_cond} and \eqref{Eq:quadratic_span}, it is easy to see that this model admits an AQED code, but not an AQEC code. Following the analysis in  Ref.~\cite{breuer_book} (section 10.1), for an initial state in the codespace constructed by Lemma~\ref{lem:Error_detection} and evolving under the Jaynes-Cummings model, the amplitude $\beta(t)$ for the qubit to be in the state $|1\rangle$ is the solution to the integro-differential equation 
\begin{equation*}
\label{eq:integro_diff}
    \dot{\beta}(t) + i  \omega_0 \beta(t) + \int_{0}^t~ d\tau  ~ \int_0^\infty d\omega   J(\omega) e^{-i \omega (t - \tau)}  \beta(\tau)=0
    .
\end{equation*}
 In the above equation, $J(\omega) = \sum_k |g_k|^2 \delta(\omega - \omega_k)$ is the spectral density associated with the environment. 
\begin{figure*}[htbp!]
\includegraphics[width=\linewidth]{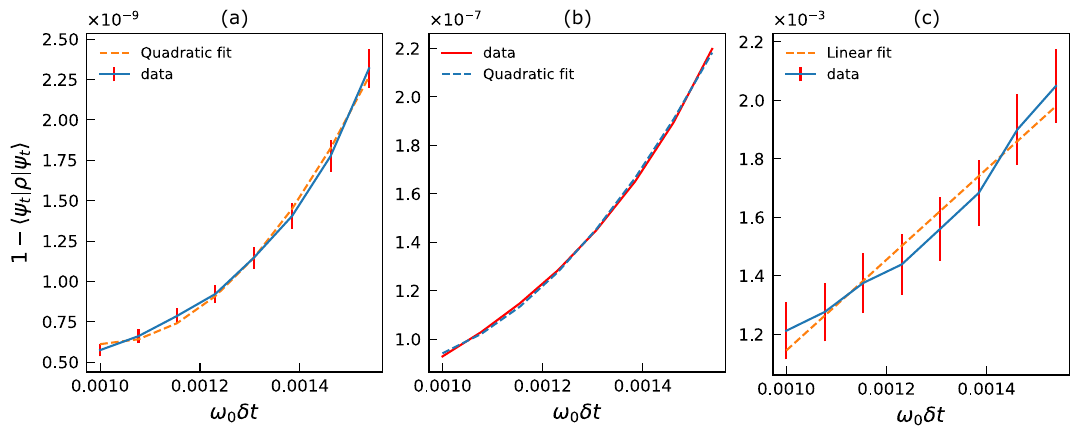}
\caption{Infidelity of state at time $\omega_0 t = 10$ ($\rho(t)$) with the exact evolution ($\ket{\psi_t}$) vs $\delta t$. (a) plots the AQEC evolution, (b) the DD evolution and (c) plots the AQED evolution. The red bars on (a) and (c) are standard error $\sigma/\sqrt{N}$, where $\sigma$ is the standard deviation for $N = 10000$ trajectories.  We plot linear and quadratic fits generated using the \textit{polyfit} function in \textit{numpy} library of python.}
\label{fig:scaling_plot}
\end{figure*}
For an appropriate continuum limit of the bath with a polynomial cutoff, the Jaynes-Cummings spectral density becomes~\cite{breuer_book}
\begin{equation}
    J(\omega) = \frac{1}{2 \pi} \frac{\gamma_0 \lambda^2}{(\omega_0 - \Delta - \omega)^2 + \lambda^2},
\end{equation}
where $\Delta$ is the detuning of the cavity frequency from atomic resonance, the parameter $\lambda$ controls the spectral width of the coupling, and the parameter $\gamma_{0}$ controls the rate at which the atom evolves due to interaction with the bosonic modes. In App.~\ref{ap:spon_decay}, we  discuss the behavior of models of this type in more detail.

The application of the AQED protocol to this model results in one of a set of pure states $\{|\psi_{j}\rangle\}$ being produced. In particular, the state $|\psi_{j}\rangle$ produced depends on the measurement record. Each of these states is produced with an associated probability $p_{j}$. We define the average QFI for this ensemble to be
\begin{equation}
   F_{\mathrm{avg}}(t) = \sum_{i = 1}^m p_i \mathcal{F} \lf \ketbra{\psi_i}{\psi_i}\ri.  
\end{equation}
In Fig.~\ref{fig:combined_image}, we plot the average QFI  as a function of time for various measurement timescales $\delta t$. From the figure, we can immediately draw the following conclusions. First, as $\omega_0\delta t$ approaches $0$, we see that the QFI scales quadratically with time $t$, signifying the recovery of Heisenberg limit thereby reinforcing Theorem \ref{Thm:AQED_zeno}. Second, the QFI for $\omega_0\delta t  > 0.1$ performs worse than the uninterrupted evolution. This shows us that the timescale of $\delta t$ is an important factor in the performance of AQED codes. The AQED faithfully gives a QFI that grows like $t^2$ if $\delta t$ is the fastest timescale, i.e, $\delta t \ll \tau_B, \tau_R$, where $\tau_B=\lambda^{-1}$ is the reservoir correlation time and $\tau_R=\gamma_{0}^{-1}$ is the timescale on which the system changes.

Second, we consider a system that admits an AQEC code and run it using both AQEC and AQED procedures to compare their performance. In the former, we construct the recovery operator $\mathcal{R}_{\tilde{\mathcal{C}}}$ as in Eq.~\eqref{EQ:Errorspace_recovery} and for the latter we reset it to the initial state whenever an error is detected. For this simulation, we consider the case of perpendicular dephasing under the Hamiltonian 
\begin{equation}
    H = \omega_0\sigma_{z,S} + \sum_j g_{env} (j) \sigma_{x,S} \otimes \sigma_{x,E}^{(j)} +   \Omega_{env}(j) \sigma_{z,E}^{(j)} .
\end{equation}
The probe, ancilla and environmental modes are all modeled as qubits. In the results presented, we consider $5$ environmental modes with the environment frequencies chosen randomly (see App.~\ref{ap:Perp_dephasing_sim} for details). In Fig.~\ref{fig:scaling_plot}, we plot the infidelity of state undergoing noisy evolution $\rho(t)$ at time $t$ with the state undergoing noiseless evolution $\ket{\psi_t}$ at time $t$, for evolutions with different measurement timescales $\delta t$. 

We choose the infidelity between the two states, as it is an easier to calculate surrogate for the absolute value of the difference between the QFI of the two states. In particular, if frequency encoding occurs via a purely unitary operation, then, as shown in Ref.~\cite{Augusiak_2016} Theorem 1, the absolute value of difference between the two QFIs can be bounded by the infidelity as  
\begin{equation}
    \left| \mathcal{F}(|\psi\rangle\langle\psi|) - \mathcal{F}(|\phi\rangle\langle\phi|) \right| \leq 32t^{2} \sqrt{1 - |\!\braket{\phi|\psi}\!|^2}. 
\end{equation}
While the situation we consider here does not exactly satisfy this condition, this suggests that if the fidelity is close to one then the states will yield almost the same QFI. From the figure, it is clear that the infidelity scales as $O(\delta t^2)$ for AQEC and DD schemes while it scales as $O(\delta t)$ for the AQED scheme, as the fits match the numerical curve, up to the standard error. \par 
In Fig.~\ref{fig:DvC} we plot the fidelity with the initial state vs time for the three different evolutions. The plot in the left is for a measurement rate of $\delta t = 0.1$ and the one on the right for $\delta t = 0.001$  for a parameter regime that is not captured by our leading order analysis. From the figures it is clear that the AQEC and DD perform much better than AQED, even under this regime.  

\section{Conclusions}  

We have analyzed the performance of three types of schemes intended to protect frequency estimation protocols from the harmful effects of noise. In particular, we have studied their performance as a function of the control timescale. All of the considered protocols are known to yield Heisenberg-limited scaling in the limit that the controls can be applied arbitrarily fast. Despite this fact, major differences in performance emerge when the time between control operations, which we denote by $\delta t$, is made finite. In particular,  we analytically show that approximate quantum error detection strategies are guaranteed to yield Heisenberg scaling only until a time that grows as $\delta t^{-1}$ while both approximate quantum error correction based strategies and dynamical decoupling based strategies are guaranteed to yield Heisenberg scaling till a time that grows as $\delta t^{-2}$. This amounts to a major difference when $\delta t$ is small. 

\begin{figure}
\centering
\includegraphics[width=\linewidth]{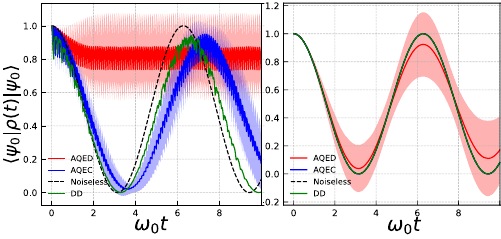}
\caption{The mean Rabi oscillation over different trajectories is plotted along with the variance for two different $\delta t$ for AQEC and AQED code. The left is for $\omega_0\delta t = 0.1$ and the right is for $\omega_0 \delta t = 0.001$. The shaded regions represent the standard deviations of the respective quantities.}
\label{fig:DvC}
\end{figure}

The physical origin of the weaker noise protection provided by the quantum error detection based strategies is a reliance on the quantum Zeno effect. In particular, these strategies rely on the fact that in the limit that $\delta t$ tends to zero all syndrome measurements yield the no error outcome. This is a generic feature of unitary evolution where after a short evolution for time $\delta t$, the probability of projecting via measurement onto a subspace orthogonal to the initial state is at most $\delta t^{2}$. However, if many measurements are performed after many short time evolutions with finite $\delta t$, the probability of at least one error being detected becomes substantial. It is often the case for this type of protocol, that an error entails the complete erasure of all information about the frequency currently encoded in the parameter. This type of catastrophic failure generally leaves no recourse but to start the sensing protocol over again from the beginning. 

Many of our results apply rigorously to the case of environments associated with infinite-dimensional Hilbert spaces. While this substantially increases the technical difficulty of the results, this effort is worthwhile because it creates a bridge between the non-Markovian setting considered here and Markovian setting that dominates the literature. In particular, truly Markovian quantum dynamics are possible only for a system coupled to an environment with an infinite-dimensional Hilbert space. In the fully Markovian setting, the conditions under which Heisenberg-limited scaling can be recovered, even with fast controls, are more stringent than in the non-Markovian setting. Physically, this is because the ``fast controls" in the Markovian setting  still operate more slowly than the fast timescales associated with the environment while ``fast controls" in the non-Markovian setting are faster than all other timescales in the problem. The promotion of $\delta t$ to a tunable parameter, as done in this work, allows for a certain type of interpolation from these two regimes. We expect that in the future studying noisy metrology in the presence of this type of tunable interpolation will continue to yield further insights. 

It is worth dwelling on a slight asymmetry between the conditions for quantum error correction based strategies to obtain Heisenberg like scaling up till time $\delta t^{-2}$. The usual span type condition appears but new conditions also appear. These new conditions are related to the parts of $G_{\perp}$ that link codespace and the error space. These conditions are automatically satisfied if $G_{\perp}$ is block diagonal between these two subspaces, i.e., if $G_{\perp}$ is itself a logical operator. However, this is not the only way to satisfy these conditions. It would be interesting to develop a more in depth theory of which types of signals and errors satisfy these conditions. It is also possible that the use of an alternative family of codes allows this condition to be relaxed. Finally, a theory of which conditions emerge if Heisenberg scaling until time $\delta t^{m}$ is required may be useful for use with ``high distance" quantum error correcting codes.

As $\delta t$ is allowed to grow so that we approach the Markovian regime, the quantum error correction protocols we have discussed here should, at least within some regime, tend towards their Markovian counterparts. In this limit, they yield Heisenberg scaling, at least up till some time, in a similar manner to that discussed here. On the other hand, dynamical decoupling protocols typically do not find use in the fully Markovian regime. Interestingly, error detection protocols occupy somewhat of a middle ground between these two extremes. They do not yield Heisenberg-limited scaling but they do find metrological use. For example, in the extreme limit where only a single round of error detection can be performed prior to measurement, data processing inequalities rule out the resulting protocol yielding Heisenberg-limited scaling but the error detection can still be used as a heuristic method to obtain an optimal measurement by flagging samples that contain no information about the frequency being measured. In this sense, the error detection protocols we discuss here are related to error mitigation and post-selection based strategies in the Markovian limit. 

In the past, most studies of quantum error corrected metrology have focused on the $\delta t$ goes to zero limit. This paradigm is extremely important as it often serves to establish the fundamental bounds on the estimation error that it is in principle possible to achieve. However, as we enter the age of early fault tolerant quantum devices it is worth examining the operation of quantum error correction protocols under less ideal circumstances. For example, we should consider, among other effects, finite error correction cycle time, as we have done here, the occurrence of errors during the syndrome extraction and correction procedure, and the effect of signal accumulation during the error correction procedure. All of these are certain to have an impact on the performance of protocols in any use case. 

\section*{Data availability}

Data generated and analyzed during the current study are available from the corresponding author upon reasonable request

\section*{Acknowledgment}
We thank Manuel Mu{\~n}oz-Arias, Francisco Riberi, Ivan Deutsch, Marco Rodriguez Garcia, and Sisi Zhou for fruitful discussions. We thank Akimasa Miyake for his support and participation during the early stages of this project. This work is supported by the National Science Foundation QLCI Q-SEnSE (Grant No. OMA-2016244).
\bibliography{ref}
\appendix
\onecolumngrid

\section{Proof of Theorem 1}
\label{ap:Proof_theorem_1}

\subsection{Definitions and notations}
Since we use notions of an infinite-dimensional Hilbert space, we encounter operators that might not be \textit{bounded}. For the sake of completeness and for the convenience of readers, we will review the basic definitions and results in this subsection. Interested readers who wish to learn more are directed to the books written by Reed and Simon \cite{reed_simon}, Weidmann \cite{weidmann_linear_1980}, and Schmudgen \cite{schmudgen_unbounded_2012}. Let $\mathcal{H}$ be a separable Hilbert space, i.e, it has a countable basis. Additionally, we denote the domain of an operator $T$ by $D(T)$.
\begin{define}[Linear operator]
    An operator $T: D(T) \to \mathcal{H}$ is called a linear operator, if for every $x,y \in D(T)$ and $\alpha, \beta \in \mathds{C}$, we have
    \begin{equation}
        T(\alpha x + \beta y) = \alpha T(x) + \beta T(y).
    \end{equation}
    Furthermore, if $D(T)$ is dense in $\mathcal{H}$, then we call $T$ densely defined.
\end{define}
Any linear operator is defined along with its domain $D(T)$. The domain of an unbounded operator is very crucial. Operators with the same form but different domains behave very differently and might have completely different properties. To illustrate this point, consider the differential operator $\mathcal{T} = -\frac{ d^2}{dx^2}$ on two different domains, 
\begin{equation}
    D(\mathcal{T}_D) = \left\{ f \in H^2([0,1]) | f(0) = f(1) = 0 \right\}, \quad D(\mathcal{T}_P) = \left\{ f \in H^2([0,1]) | f(0) = f(1) \text{ and } f'(0) = f'(1) \right\},
\end{equation}
where $H^2{([0,1])} \subset L^2([0,1])$ is the set of twice (weakly) differentiable functions. $L^2([0,1])$ is the set of square integrable functions on $[0,1]$. It can be shown that both these domains are dense in $L^2([0,1])$. Solving the eigenvalue equation $f'' = -\lambda f$, we see that the domains $D(\mathcal{T}_D)$ and $D(\mathcal{T}_P)$ have different spectra. In fact, $D(\mathcal{T}_D)$ does not have a zero eigenvalue, while $D(\mathcal{T}_P)$ has one (corresponding to a non-zero constant function). Different domains having different solutions should not be surprising, as they correspond to different physical settings. For instance, the above example can be considered a free particle in an infinite square well (for $D(\mathcal{T}_D)$) and a ring (for $D(\mathcal{T}_P)$). We will distinguish the operator $\mathcal{T}$ defined on the two different domains by representing it as $\mathcal{T}_D$ and $\mathcal{T}_P$.

Throughout this article, we will assume that all operators are densely defined, that is to say that its domain $D(T)$ is a dense subset of the Hilbert space $\mathcal{H}$. Because of the domain, the notions of equality of operators and restrictions become a little bit more involved. 
\begin{define}[Equality and restrictions]
    Two linear operators $T,S$ are said to be equal iff $D(S) = D(T)$ and $Sx = Tx$ for all $x \in D(T)$. Furthermore we will say that $T$ is an extension of $S$ (or $S$ is a restriction of $T$) if $D(S) \subset D(T)$ and 
    \begin{equation}
    Sx = Tx ~~ \forall x \in D(S). 
    \end{equation}
    This is denoted by $S = T|_{D(S)}.$
\end{define}
In particular, note that two operators $T$ and $S$ are considered equal only if their domains also coincide. Another important property is the closedness of linear operators, which are defined using the notion of a graph. 
\begin{define}[Graph]
    The \textit{graph} of a linear operator $T$ is given by, 
    \begin{equation}
        \mathcal{G}(T) = \big\{ (x, Tx) \in \mathcal{H} \times \mathcal{H}  | x \in D(T)\big\}
    \end{equation}
\end{define}
It is easy see that any operator is uniquely identified by its graph $\mathcal{G}$ as it lists the entire domain and the mapping. This is used to define norms on the space of operators and the notion of closedness and closability. 
\begin{define}[Closed and Closable]
    An operator $T$ is called \textit{closed} if its graph $\mathcal{G}(T)$ is a closed subset of $H \times H$. An operator $T$ is called \textit{closable} if there exists a closed linear operator $S$ such that, 
    \begin{equation}
        T = S|_{D(T)}.
    \end{equation}
\end{define}
The notion of closed operators can be thought of as a generalization of bounded operators. In fact, the notion of a closed operator is necessary to generalize the concept of a spectrum. We now state a proposition from Ref.~\cite{schmudgen_unbounded_2012}.
\begin{prop}
\label{prop:Closure_property}
    The following statements are equivalent.
    \begin{enumerate}
        \item $T$ is closed
        \item If $(x_n)_{n \in \mathds{N}}$ is a sequence of vectors $x_n \in D(T)$ such that $\lim_{n \to \infty} x_n = x$ and $\lim_{n \to \infty} Tx_n = y$, then $x \in D(T)$ and $Tx = y$.
    \end{enumerate}
\end{prop}
We denote the closure of a set $\mathcal{S}$ by $\bar{\mathcal{S}}$. We also define the closure $\overline{T}$ of a closable operator $T$. 
\begin{define}[Closure]
    For a closable operator $T$ with graph $\mathcal{G}(T)$, we define the closure $\overline{T}$ as the unique linear operator with graph $\overline{\mathcal{G}(T)}$.
\end{define}
From the above definition, $D(\overline{T})$ is the set of vectors $x \in \mathcal{H}$, for which there exists a sequence $(x_n)_{n \in \mathds{N}} \in D(T)$ which converges to $x$ such that $(T(x_n))_{n \in \mathds{N}}$ converges in $\mathcal{H}$. We will now define the notions of a symmetric operator, adjoint and self-adjoint operator. 
\begin{define}[Symmetric operator]
    A densely defined operator $T$ is called symmetric if for all $x,y \in D(T)$, we have 
    \begin{equation}
        \braket{x,Ty} = \braket{Tx,y},
    \end{equation}
    where $\braket{.,.}$ is the inner product of the Hilbert space. 
\end{define}
\begin{define}[Adjoint operator]
    Define the domain of the adjoint operator $D(T^{\dagger})$ to be
    \begin{equation}
        D(T^{\dagger}) = \big \{ y \in \mathcal{H} \big|  \exists u \in \mathcal{H} \text{ such that }\forall x \in D(T),  \braket{Tx,y} = \braket{x,u} \big\},
    \end{equation}
    where $\braket{.,.}$ is the inner product of the Hilbert space. Then the adjoint of any densely defined operator $T$ is given as the linear map, 
    \begin{equation}
        T^{\dagger} y = u.
    \end{equation}
\end{define}
 The denseness of $D(T)$ ensures that $u$ is unique. An operator $T$ is called self-adjoint if $T = T^{\dagger}$. Note that $T$ being self-adjoint is a more stringent condition than it being symmetric. For any  symmetric $T$, we can immediately see that $D(T) \subset D(T^{\dagger})$. However, for $T$ to be self-adjoint, we require $D(T) = D(T^\dagger)$. Let us illustrate this point more carefully using the example $\mathcal{T}$ defined on the domain, 
\begin{equation}
    D(\mathcal{T}_{\text{sym}}) = \{f \in H^2([0,1]) | f(0) = f(1) = 0, f'(0) = f'(1) = 0\}
\end{equation}
For this domain, it is easy to show that for any $f \in H^2([0,1])$, 
\begin{equation}
    \braket{\mathcal{T}_{\text{sym}} g,f} =  - \int_0^1 dx \frac{d^2g^*}{d x^2} f =  - \int_0^1 dx g^*\frac{d^2f}{d x^2} = \braket{g,\mathcal{T} f}, \quad \forall g \in   D(\mathcal{T}_{\text{sym}}).
\end{equation}
This can be seen by doing integration by parts, which will leave a boundary term $[f(x) g^{*'}(x) - f^{'}(x) g^{*}(x)]_0^1$, which vanishes as $g(x), g'(x)$ vanish at the boundary. This immediately implies that $\mathcal{T}_{\text{sym}}$ is symmetric and  $D(\mathcal{T}^\dagger_{\text{sym}}) = H^2([0,1]) \supset D(\mathcal{T}_{\text{sym}})$. i.e, it is not self-adjoint. 

If $T$ is bounded, then the two definitions of symmetric and self-adjoint coincide, as operators are everywhere defined, which is why we do not care for the difference between them in any finite-dimensional case. 
In general a densely defined symmetric operator $T$ we can have either a unique, or many or no self-adjoint extensions. For example $\mathcal{T}_D$ and $\mathcal{T}_P$ are both valid self-adjoint extensions of $\mathcal{T}_{\text{sym}}$. To this end, we define the notion of an essentially self-adjoint operator. 

\begin{define}[Essentially self-adjoint]
    A dense symmetric linear operator $T$ is called essentially self-adjoint if it's closure $\overline{T}$ is self-adjoint. 
\end{define}
The being essentially self-adjoint effectively ``restricts" the possible self-adjoint extensions of $T$, as essentially self-adjoint operators can be shown to have a unique self-adjoint extension. To see this note that, since the closure is the closed extension of $T$ with the smallest graph, any self-adjoint extension $S$ of $T$,  $T = S|_{D(T)}$,  is also an extension of the closure. Thus (see Proposition 1.6 of Ref.~\cite{schmudgen_unbounded_2012}), 
\begin{equation}
    \bar{T} = S|_{D(\bar{T})} \implies S^{\dagger} = \bar{T}^\dagger|_{D(S^\dagger)}.
\end{equation}
Self-adjointness of $S$ and $\bar{T}$ implies $S = \bar{T}$.

\subsection{Proof of the Theorem \ref{Thm:AQED_zeno}}
In this section, we rigorously prove Thm.~\ref{Thm:AQED_zeno}. For the sake of completeness, we restate our setup. We consider a finite-dimensional quantum system described by the Hilbert space $\mathcal{H}_S$ of dimension $d < \infty$, coupled to an environment of arbitrary (possibly infinite) dimension defined on a separable Hilbert space $\mathcal{H}_E$. The joint dynamics are governed by the total Hamiltonian,   
\begin{equation}
     \label{ApEq:Setup_Hamiltonian}
     H_{\text{tot}} = \omega_0 G \otimes \mathds{1} + \sum_{k = 1}^{K} g_{k}E_k \otimes B_k + g_{0}\mathds{1} \otimes B_0 + \alpha \mathds{1} \otimes \mathds{1},
\end{equation}
where $G,\{E_k\} \in \mathcal{B}(\mathcal{H}_S)$ are Hermitian operators acting on $\mathcal{H}_S$, and $B_k, B_0$ are (possibly unbounded)  operators on $\mathcal{H}_E$. In particular, the operators $E_k$ are assumed to have a norm $\|E_k\| = 1$ and $g_{k}$ is a real valued quantity with units of frequency.  The Hamiltonian $H_{\text{tot}}$ is assumed to be self-adjoint in the composite Hilbert space. Furthermore, we assume without loss of generality that $g_0, g_k$ are non-zero.  We refer to $G$ as the signal generator and to the $\{E_{k}\}$ as the error generators. In order to prove Theorem.~\ref{Thm:AQED_zeno}, we first need Lemma.~\ref{lem:Error_detection}. As stated in the main text, this lemma was first proved in Ref.~\cite{zhou_zhang_etal, mann_zhou_laflamme} but we provide the proof here for the sake of completeness.  Let $G$ be a Hermitian operator and $\{E_k| k = 1,\dots, K\}$ be an orthonormal set of traceless operators such that 
\begin{equation}
\label{Apeq:hnls}
    G \notin \text{span} \{E_k, \mathds{1}\}.
\end{equation}
Define the operator $G_\perp$ as, 
    \begin{equation}
    \label{ApEq:G_perp}
        G_\perp = G - \sum_{k = 1}^K \operatorname{Tr} \lf G E_k\ri E_k - \operatorname{Tr}(G) \mathds{1}.
    \end{equation}
The operator $\sum_{k = 1}^K g_k E_k \otimes B_k$ being Hermitian implies that for every $k$, there exists a $k'$ such that $E_{k'} = E^\dagger_k$.  Thus, $G_{\perp}$ is implicitly Hermitian. Furthermore, since $G \notin \text{span} \{E_k, \mathds{1}\}$, we see that $G_\perp$ is a traceless and non-trivial. Hence, using the spectral theorem, we can write, 
\begin{equation}
    \label{ApEq:G_perp_spec_decomp}
    G_\perp = \rho_+ - \rho_-, 
\end{equation}
for positive operators $\rho_A, \rho_B$ with equal trace.
\begin{lemma*}[\textbf{\ref{lem:Error_detection}}]
\label{Aplem:Error_detection}
    Let $G$ satisfy Eq.~\eqref{Apeq:hnls}, $G_\perp$ be defined as in Eq.~\eqref{ApEq:G_perp}, and $\rho_{+}$ and $\rho_{-}$ be positive operators that satisfy Eq.~\eqref{ApEq:G_perp_spec_decomp}. Define the states $\ket{\psi_+}, \ket{\psi_-}  \in \mathcal{H}_S \otimes \mathcal{H}_S$ to be purifications of $\rho_+$ and $\rho_-$. Now, define the projection operator 
    \begin{equation}
        \Pi_C = \ketbra{\psi_+}{\psi_+} + \ketbra{\psi_-}{\psi_-}.
    \end{equation}
    The projection operator $\Pi_C$  satisfies 
    \begin{equation}
        \label{Apeq:Lemma_constants}
        \Pi_C G \Pi_C \not\propto \Pi_C, \quad \Pi_C E_k \Pi_C = \lambda_k \Pi_C ~ \forall k, \lambda_k \in \mathds{C}. 
    \end{equation}
\end{lemma*}
 \begin{proof}
     We begin noting that $\rho_A$ and $\rho_B$ can be decomposed as   
 \begin{equation}
 \label{eq:spec_decomp_eval}
     \rho_+ = \sum_j \lambda_j \ketbra{\phi_j}{\phi_j}, \quad \rho_- = \sum_j \gamma_j \ketbra{\psi_j}{\psi_j},
 \end{equation}
 where $\ket{\psi_j}$ and $\ket{\phi_j}$ satisfy $\braket{\psi_j|\phi_k} = 0$ for all $j,k$. From this decomposition, we can write the states $\ket{\psi_{\pm}}$ as
 \begin{equation}
     \ket{\psi_+} = \sum_{j} \sqrt{\lambda_j}\ket{\phi_j}\ket{\phi_j},  \text{and } \ket{\psi_-} = \sum_{j} \sqrt{\gamma_j}\ket{\psi_j}\ket{\psi_j}.
 \end{equation}
 Using Eq. \eqref{eq:spec_decomp_eval}, we immediately conclude the following, 
\begin{equation}
    \label{eq:conseq_spec}
    \operatorname{Tr}_A \lf \ket{\psi_+} \bra{\psi_-} \ri = 0,
\end{equation}
where the trace is taken over the ancilla. Using $\operatorname{Tr}(G_\perp E_k) = 0$ and $\operatorname{Tr}_A\lf \ketbra{\psi_+}{\psi_+} - \ketbra{\psi_-}{\psi_-} \ri = G_\perp$, we see that 
\begin{equation}
    \label{eq:conseq_spec_1}
    \braket{\psi_+|\mathds{1}_A \otimes E_k |\psi_+} = \braket{\psi_-|\mathds{1}_A \otimes E_k |\psi_-}.
\end{equation}
Now consider the following chain of equations, 
\begin{align}
\Pi_C  \mathds{1}_A \otimes E_k  \,\Pi_C 
&= \ketbra{\psi_+}{\psi_+} 
    \braket{\psi_+| \mathds{1}_A \otimes E_k  |\psi_+}  + \ketbra{\psi_-}{\psi_-} 
    \braket{\psi_-| \mathds{1}_A \otimes E_k  |\psi_-} \nonumber\\
&\qquad  \qquad + \ketbra{\psi_+}{\psi_-} 
    \braket{\psi_+| \mathds{1}_A \otimes E_k  |\psi_-}  + \ketbra{\psi_-}{\psi_+} 
    \braket{\psi_-| \mathds{1}_A \otimes E_k  |\psi_+}, \nonumber\\
&= \operatorname{Tr}(\rho_+ E_k) 
    \left[ \ketbra{\psi_+}{\psi_+} 
         + \ketbra{\psi_-}{\psi_-} \right].
\end{align}
In the last line in the above chain of equalities, we have used Eqs.~\eqref{eq:conseq_spec} and~\eqref{eq:conseq_spec_1}. Similarly, we can expand $G$ into components that are perpendicular and parallel to $\text{span}\{E_k, \mathds{1}\}$. Since $G_\parallel$ is orthogonal to $G_\perp$, it will be projected trivially by the operator $\Pi_C$. Hence 
\begin{equation}
    \Pi_C  \mathds{1}_A \otimes G \Pi_C  = \Pi_C  \lf   \mathds{1}_A \otimes G_\perp +  \mathds{1}_A \otimes G_\parallel  \ri \Pi_C = \Pi_C \lf  \mathds{1}_A \otimes G_\perp  + \lambda \ri\Pi_C  \not\propto \Pi_C, \qquad \lambda \in \mathds{R }. 
\end{equation}
 \end{proof}

Using Lemma~\ref{lem:Error_detection}, we will now prove Theorem~\ref{Thm:AQED_zeno}. We restate the theorem for the sake of clarity. 
\begin{thm*}[\textbf{\ref{Thm:AQED_zeno}}]
  \label{ApThm:AQED_zeno}
Let $\mathcal{{H}}_S$ be a finite-dimensional Hilbert space for the system and  $\mathcal{H}_E$ a separable Hilbert space for the environment. Consider the joint evolution on $\mathcal{{H}}_S \otimes \mathcal{{H}}_E$ generated by the Hamiltonian in Eq.~\eqref{Eq:Setup_Hamiltonian}.  
 Assume:
 \begin{enumerate}
     \item  The Hamiltonian $H_{\text{tot}}$ is non-negative,
     \item  One has access to noiseless ancillas and the ability to perform arbitrary unitary controls and measurements on the system and ancilla.
 \end{enumerate}
Under these assumptions, we can achieve a Heisenberg-limited scaling of the quantum Fisher information of the parameter $\omega_0$, i.e  
 \begin{equation}
     \lim_{t \to \infty}\lim_{\delta t \to 0} \mathcal{F}(t) = \Theta(t^2) ,
 \end{equation}
provided that the signal generator $G$ satisfies the not-in-span condition 
   \begin{equation}
       G \notin \text{span}_{\mathbb{H}}\lf E_k , \mathds{1}\ri ~ \mathrm{for~all} ~~ k \in \{1, \dots, K\},
       \label{eq:span_cond_ap}
   \end{equation}
   and that the symmetric operator $P_0 H_{tot} P_0$ is essentially adjoint. Here $P_0 = \Pi_C \otimes 1_B$ is a projector, where $\Pi_C$ is constructed as per Lemma \ref{lem:Error_detection}.
\end{thm*}
 \begin{proof}
 We will define our codespace as the two-dimensional subspace of $\mathcal{H}_S \otimes \mathcal{H}_S$ spanned by $\{\ket{\psi_A}, \ket{\psi_B}\}$,defined in Lemma~\ref{lem:Error_detection}. This is because projecting into the codespace trivializes the errors without trivializing the signal. In the following, we will make repeated two-outcome measurements of $\cal{M}$  on the system and ancilla, where $\cal{M}$ is defined by the projection operators $\{\Pi_C, \Pi_{\tilde{C}}\}$,. The operators $\{\Pi_C, \Pi_{\tilde{C}}\}$ are defined as the projections into the codespace and its orthogonal subspace, respectively. Since we keep track of the environment degrees of freedom, the effective measurement on the total Hilbert space is given by the operators
\begin{equation}
    P_0 = \Pi_C \otimes \mathds{1}_B, \quad P_1 = \Pi_{\tilde{C}} \otimes \mathds{1}_B.
\end{equation}

Note that by construction, $[P_0, H_{tot}] \neq 0.$. Let $t$ be the total evolution time. We start with a state $\rho_S \otimes \ketbra{0}{0}$ and let it evolve under the total Hamiltonian. The initial system state is assumed to be in the codespace, $\Pi_C \rho_S \Pi_C = \rho_S$. We will make a series of instantaneous measurements and recovery operations at times $\left\{ t/n, 2t/n, \dots, (n-1)t/n, t\right\}$, where the recovery operation is given as 
\begin{equation}
    \label{Eq:recovery}
    R(\rho) = P_0 \rho P_0 + R_E({P_{1}} \rho {P_{1}}), 
\end{equation}
where $R_E(\rho) = \operatorname{Tr}(\rho) \rho_S \otimes \ketbra{0}{0}$ is the reset to the initial state. Note that this recovery operator assumes the ability to reset the entire Hilbert space. Even though we have a strict recovery operation, $R_E$, we will see that the exact structure of this recovery operation is not important for our purpose. Hence, this structure is only for mathematical convenience. However, from a physical point of view, this can be viewed as throwing away the experiment and starting again. Define the unitary $U$ on $\mathcal{H}_S \otimes \mathcal{H}_S \otimes \mathcal{H}_E$ as 
\begin{equation}
    U(t) = \exp \lf -i t \mathds{1}_A \otimes H_{\text{tot}} \ri.
\end{equation}
Denoting by $\rho(m,t)$ the state of the system after $m$ non-selective measurements and defining the superoperator $\mathcal{U}(\rho) = U(t/n) \rho U^\dagger(t/n)$, we get the expression 
\begin{equation}
    \rho(m,t) = (R\circ\mathcal{U})^m (\rho_s \otimes \ketbra{0}{0}).
 \end{equation}
Substituting in Eq.~\eqref{Eq:recovery}, we get the following expression for the final state after $n$ measurements, 
\begin{align}
    &\rho(n,t) = \hat{T}_n(t) \rho_s \otimes \ketbra{0}{0} \hat{T}^\dagger_n(t)  + \mathcal{D} (\rho_s \otimes \ketbra{0}0{}), 
\end{align}
where $\mathcal{D}(\odot)$ denotes the channel that keeps track of trajectories where the state was projected onto the error space at least at one time instance. The exact structure of $\mathcal{D}$ depends on the recovery $R_E$, but it is not too important, as mentioned before. The crucial part to note is that the term $\hat{T}_n(t) \rho_s \otimes \ketbra{0}{0} \hat{T}^\dagger_n(t)$ is the trajectory where the state was projected onto the codespace at all times. The operator $T_n(t)$ is thus,  
\begin{equation}
    T_n\lf t \ri := \lf P_0 U\lf \frac{t}{n} \ri P_0\ri^n.
\end{equation}
Let $p_n(t)$ denote the probability of the trajectory where the state was never projected out of the codespace 
\begin{equation}
    p_n(t) = \operatorname{Tr} \lf T^\dagger_n\lf n;t \ri T_n\lf n;t \ri \rho_s \otimes \ketbra{0}{0}\ri.
\end{equation}
We will now use the Zeno effect \cite{misra_sudarshan} to show that $\lim_{n \to \infty} p_{n(t)}  = 1$. This means that in the limit of infinite measurements, the system will always remain in the codespace. Furthermore, we will show that the resulting dynamics in the codespace is generated by a Hamiltonian that is separable between the system and the environment. This results in a pure unitary evolution that leads to the Heisenberg limit.

Since $H_{\text{tot}}$ is a non-negative self-adjoint operator, it is bounded from below. Assume that the following limits exist 
\begin{align}
    \lim_{n \to \infty} T_n(n,t) &= T(t), ~~ \forall t \in \mathds{R}, \label{Eq:Limits_a} \\
    \lim_{t \to 0} T_n(n,t) &= P_{0}, ~~ \forall n \in \mathds{N} \label{Eq:Limits_b}.
\end{align}

Then, we satisfy all the conditions to apply Theorem 1 in Ref.~\cite{misra_sudarshan} and hence $T(t)$ forms a strongly continuous semigroup with the functional equation $T(t)T(s) = T(t+s)$ for all $t,s \in \mathds{R}$ and $T^\dagger(t) = T(-t)$. Combining these with Eq.~\eqref{Eq:Limits_b}, we get 
\begin{equation}
    \lim_{n \to \infty} p_{n(t)}  = \operatorname{Tr} \lf  T(t)^\dagger  T(t) \rho_s \otimes \ketbra{0}{0} \ri  = \operatorname{Tr} \lf T(0) \rho_s \otimes \ketbra{0}{0} \ri  = 1. 
\end{equation}
We still need to prove when the limits in Eqs.~\eqref{Eq:Limits_a} and~\eqref{Eq:Limits_b} exist. If $H_{\text{tot}}$ is bounded, then we can prove that the limits always exist \cite{facchi_pascazio}. Define $\tilde{H} = \mathds{1}_A \otimes H_{\text{tot}}$ for notational convenience and note that it is bounded. This is proven easily as 
\begin{align}
    \lim_{n \to \infty} T_n(t) &=  \lim_{n \to \infty} \lf \sum_{m = 0}^{\infty}\frac{(-it)^m}{ n^m m!} P_0 \tilde{H}^m P_0  \ri^{n},  \label{eq:lim_1}\\
    &= \lim_{n \to \infty} \lf P_0  - \frac{it}{n} P_0\tilde{H} P_0 \ri^n\nonumber  \\
    &\quad \quad + O\lf \lf \frac{t^2}{n} \ri  \|\tilde{H}\|^{2} e^{\|H\| t}\ri  \label{eq:lim_2}\\
    &= \lim_{n \to \infty} \lf \mathds{1}  - \frac{it}{n} P_0\tilde{H} P_0  \ri^n P_0  \label{eq:lim_3}\\
    &= \exp{\lf -it P_0 \tilde{H} P_0 \ri} P_0 \label{eq:lim_4}.   
\end{align}
In line \eqref{eq:lim_3} we have used the property of projectors $P_0^n = P_0$ and in line \eqref{eq:lim_4} we have used the Euler limit for the exponential. Using the same lines of reasoning, we can show that Eq.~\eqref{Eq:Limits_b} is also satisfied for bounded Hamiltonians. 

If $H_{\text{tot}}$ is unbounded, the problem is not so simple. For unbounded Hamiltonians, the operator $P_0H_{\text{tot}}P_0$ acting on $\mathcal{H}_S \otimes \mathcal{H}_S \otimes \mathcal{H}_E$ need not be self-adjoint, and hence a more careful treatment is warranted. For convenience, we will use $\mathcal{H}$ to denote the combined Hilbert space $\mathcal{H}_S \otimes \mathcal{H}_S \otimes \mathcal{H}_E$ henceforth in this proof.  First, we prove the following proposition, 
\begin{prop}
\label{Apeq:Density_cond}
    The operator $\tilde{H}^{1/2} P_0$ is densely defined. 
\end{prop}
\begin{proof}
    To prove this proposition we just need to prove $D(\tilde{H}^{1/2}P_0)$ is dense in $\mathcal{H}$. We will show that this is a consequence of the essential self-adjoint-ness of $P_0 \tilde{H} P_0$. Since $P_0 \tilde{H} P_0$ is essentially self-adjoint, the domain $D(P_0 \tilde{H} P_0)$ is dense in the Hilbert space $\mathcal{H}$. Consider this domain, 
    \begin{align}
        D(P_0 \tilde{H} P_0) &= D( \tilde{H} P_0), \\
        &= D(\tilde{H}) \cap P_0 \mathcal{H} \oplus P_1 \mathcal{H}
    \end{align}   
    where $P_0 \mathcal{H} := \left\{ \ket{\psi} \in \mathcal{H}| P_0 \ket{\psi} = \ket{\psi}\right\} $ denotes the projected Hilbert space. The first line in the above equation is a consequence of the projector operators being bounded and hence defined everywhere.  Thus, $D(P_0 \tilde{H} P_0)$ being dense in $\mathcal{H}$ implies that $D(\tilde{H}) \cap P_0 \mathcal{H}$ is dense in $P_0 \mathcal{H}$. Since $\tilde{H}$ is self-adjoint, we have $D(\tilde{H}) \subset D(\tilde{H}^{1/2})$. Combining this with the above equation, we immediately see that $D(\tilde{H}^{1/2}) \cap P_0 \mathcal{H}$ is dense in $P_0 \mathcal{H}$ which implies that $D(\tilde{H}^{1/2}P_0)$ is dense in $\mathcal{H}$.
\end{proof}

Proposition~\ref{Apeq:Density_cond} implies that the operator $H_z$, defined as, 
    \begin{equation}
    H_z := \lf \tilde{H}^{\frac{1}{2}} P_0 \ri^\dagger \lf \tilde{H}^{\frac{1}{2}} P_0 \ri ,
\end{equation}
associated with the quadratic form $x \to \|\tilde{H}^{1/2} P_0 x\|^2$ with form domain $D(\tilde{H}^{1/2}P_0)$ is self-adjoint. This implies that  we satisfy the conditions in  Theorem 1.1 of Ref.~\cite{exner_ichios} (and its associated correction paper \cite{exner_correction_2021}), which is used here to show that, 
\begin{equation}
    \lim_{n \to \infty} T_n(t) \ket{\psi} = \exp \lf -i t H_z \ri P_0 \ket{\psi}
\end{equation}
We now analyze the action of $H_z$ on $P_0 \mathcal{H}$. As noted in Ref.\cite{exner_ichios}, the Hamiltonian $H_Z$ is the appropriate self-adjoint extension of the symmetric operator $P_0 \tilde{H} P_0$. However, since $P_0 \tilde{H} P_0$ is essentially self-adjoint, we can immediately show that $H_Z = \overline{P_0 \tilde{H} P_0}$. Writing this as an equality at the level of unitaries, 
\begin{equation}
    \exp \lf -i t H_z \ri P_0 \ket{\psi} = \exp \lf -i t\overline{P_0 \tilde{H} P_0} \ri P_0 \ket{\psi}.
\end{equation}
The final piece of the proof is to show that $\overline{P_0 \tilde{H} P_0}$ is decoupled. Note that using Lemma \ref{Aplem:Error_detection},  we can show that 
\begin{equation}
   P_0 \tilde{H} P_0 = \omega_0 \Pi_C G_\perp \Pi_C \otimes \mathds{1}_B + \Pi_C \otimes \sum_{k = 0}^{K} \lambda_k B_k. 
\end{equation}
where we have defined $B_0 \equiv \mathds{1}_B$ as containing terms proportional to identity.  From the properties of closure, we can immediately conclude that, $ \overline{P_0HP_0}$ is also decoupled, that is
\begin{equation}
    \overline{P_0 \tilde{H} P_0}  =\omega_0 \Pi_C G_\perp \Pi_C \otimes \mathds{1}_B + \Pi_C \otimes \overline{\sum_k \lambda_k B_k}.
\end{equation}
This can be shown in many ways, but the cleanest way to show this is to use the fact that for any densely defined closeable operator $T$, $\overline{T}^\dagger = T^\dagger$ (see Theorem 1.8 in Ref.~\cite{schmudgen_unbounded_2012}). Applying this to $\overline{P_0HP_0}$ and using the fact that it is self-adjoint, we get, 
\begin{equation}
    \overline{P_0 \tilde{H} P_0} = \lf P_0HP_0 \ri^\dagger = \lf \Pi_C G_\perp \Pi_C \otimes \mathds{1}_B + \Pi_C \otimes \sum_k \lambda_k B_k \ri^\dagger = \Pi_C G_\perp \Pi_C \otimes \mathds{1}_B + \Pi_C \otimes \overline{\sum_k \lambda_k B_k}.
\end{equation}
In deriving the last equality of the above equation, we have used the fact that $ \Pi_C G_\perp \Pi_C \otimes \mathds{1}_B$ is bounded along with  (Proposition 1.6 in Ref.~\cite{schmudgen_unbounded_2012}), 
\begin{prop}
    For a densely defined $T$ and bounded $S$ 
    \begin{equation}
        \lf S + T \ri^\dagger = S^\dagger + T^\dagger.
    \end{equation}
\end{prop}
 The final form is arrived at by noting that $\overline{P_0 \tilde{H} P_0}$ being self-adjoint imposes that $\overline{\sum_k \lambda_k B_k}$ must also be self-adjoint. From Stone's theorem, for any $\ket{\psi_0} = \ket{\psi}_{SA} \otimes \ket{\psi}_B \in D(\overline{P_0 \tilde{H} P_0}) \cap P_0 \mathcal{H}$, we have 
 \begin{equation}
     \ket{\psi(t)} := \exp \lf -i t \overline{P_0HP_0}\ri \ket{\psi_0} = \exp \lf -it\Pi_C G_\perp \Pi_C \ri  \ket{\psi}_{SA} \otimes \exp \lf -i t  \overline{\sum_k \lambda_k B_k}\ri \ket{\psi}_{B}. 
 \end{equation}
 Thus, in the limit of infinitely fast measurements the evolution of the system is decoupled from its environment.  Tracing out the environment degrees of freedom, we get a unitary evolution on the system-ancilla where the Fisher information is known to scale as $\Theta(t^2)$. Thus, 
\begin{equation}
     \lim_{\delta t \to 0} \mathcal{F}(t) = \Theta(t^2).
\end{equation}
completing the proof. 
\end{proof}

\section{Proof of Theorem 2}
\label{ap:Proof_theorem_2}
\label{ap:Recovery_operation}
Lemma \ref{lem:Error_detection} essentially gives us sufficient conditions for creating an error detecting code, which, coupled with the Zeno effect, can be used to recover Heisenberg-limited scaling. In this section, we prove that the condition is not sufficient for constructing a recovery channel, and hence standard error correction cannot be performed. This underscores the importance of the Zeno-limited measurement.  

\begin{thm*}[\textbf{\ref{Thm:Recovery_cond}}]
\label{ApThm:Recovery_cond}
    Consider the setup and joint Hamiltonian as in Eq.~\eqref{ApEq:Setup_Hamiltonian}. Define the channel $\mathcal{E}$ as
    \begin{equation}
        \mathcal{E}_{\delta t}(\rho) := \operatorname{Tr}_B  \lf \mathds{1}_A \otimes U_{\text{tot}}(\delta t) \rho \otimes \rho_B \mathds{1}_A \otimes U^\dagger_{\text{tot}}(\delta t)\ri,
    \end{equation}
    where $U_{\text{tot}}(\delta t) = \exp(-i \delta t H_{\text{tot}})$ is the unitary generated by $H_{\text{tot}}$. Suppose there exists a positive constant $C >0$, such that the $n$-point bath correlation functions satisfy 
    \begin{equation}
        \left|\operatorname{Tr}_B\lf B_{i_1} \dots B_{i_n} \rho_B \ri \right| \leq C^n.
    \end{equation}
    Then the code constructed as per Lemma \ref{lem:Error_detection} for recovering Heisenberg-limited scaling admits a recovery operation $\mathcal{R}_{\tilde{\mathcal{C}}}$ such that
    \begin{equation}
        \mathcal{R}_\mathcal{\tilde{C}} \lf \Pi_{\tilde{C}} \mathcal{E}(\rho) \Pi_{\tilde{C}} \ri \propto \rho + O(\delta t^3), \qquad \text{for} ~ \rho = \Pi_C \rho \Pi_C, 
    \end{equation}
     if 
    \begin{equation}
         \Pi_C E_i E_j \Pi_C \propto \Pi_C, \quad  \Pi_C \mathds{1}_A \otimes E_{j}  \Pi_{\tilde{C}}  \mathds{1}_A \otimes G_\perp \Pi_C \propto \Pi_C, \quad \text{and } \quad \Pi_C \mathds{1}_A \otimes G_\perp  \Pi_{\tilde{C}}  \mathds{1}_A \otimes G_\perp \Pi_C \propto \Pi_C,
    \end{equation}
    with the first two proportionalities holding for all $j$ and $k$.
\end{thm*}

\begin{proof} 
Suppose we start with a state of the form $\rho_S \otimes \rho_B$, where $\rho_S$ is in the codespace, i.e., $\Pi_C\rho_S \Pi_C = \rho_S$. Define the operators 
\begin{equation}
    \tilde{B_k} = B_k - \mean{B_k} \mathds{1}_B ,
\end{equation}
where $\mean{B_k} := \operatorname{Tr} \lf \rho_B B_k \ri$, which is finite due to the bounded correlation function property. We define these rescaled operators because we want the property $\operatorname{Tr}_B(\tilde{B}_k \rho_B) = 0$ in the proof below. Define $E_0 = \mathds{1}_S$. 
Then, we can write the total Hamiltonian as 
\begin{equation}
    H_{\text{tot}} = \omega_0 G_\perp  \otimes \mathds{1}_B + \lf \omega_0 G_\parallel + \sum_{k = 0}^K g_{k} \mean{B_k} E_k \ri \otimes \mathds{1}_B + \sum_{k = 0}^K g_{k} E_k \otimes \tilde{B_k},  
\end{equation}
where $G = G_\perp + G_\parallel$, where the perpendicular component is defined as 
\begin{equation}
  G_\perp = G - \sum_{m = 0}^{K} \operatorname{Tr}(G E_m ) E_m.  
\end{equation}
Note that the matrix $\lambda = [\lambda_{ij}]$, where $\lambda_{ij} = g_i g_j \operatorname{Tr}_B \lf  \tilde{B}^\dagger_i \tilde{B}_j \rho_B \ri$ is Hermitian (as stated before, the $g_{k}$ are real valued quantities with units of frequency) and hence diagonalizable by a unitary matrix $u = [u_{ij}]$ ($u^\dagger \lambda u$ is diagonal). Define the operators 
\begin{equation}
    F_i = \sum_{j = 0}^{K} u^*_{ij} E_j.
\end{equation}
Then, we can redefine the coupling part of the Hamiltonian as 
\begin{equation}
    H_{int} = \sum_{k = 0}^K g_k E_k \otimes \tilde{B}_k  = \sum_{j = 0}^K F_j \otimes \lf \sum_{k=0}^{K} u_{kj} g_{k}\tilde{B}_k \ri.
\end{equation}
 The operators $D_k=  \sum_j u^*_{jk} g_{j} \tilde{B}_j $ satisfy
 \begin{equation}
 \label{Eq:Condition_triv}
   \operatorname{Tr}_B \lf D^\dagger_i D_j \rho_B \ri = \operatorname{Tr}_B \lf D_i^\dagger D_i  \rho_B \ri\delta_{ij}.
 \end{equation}
Furthermore, note that 
 \begin{equation}
     \label{ApEq:Bounded_correlation_tilde}
     \left|\operatorname{Tr}_B\lf \tilde{B}_{i_1} \dots \tilde{B}_{i_n} \rho_B \ri \right| =   \left|\operatorname{Tr}_B \Bigg[ (B_{i_1} - \mean{B_{i_1}}) \dots (B_{i_n} - \mean{B_{i_n}}) \rho_B \Bigg] \right| \leq (2 C)^n.
 \end{equation}
 This can be easily seen, as the second equation above can be expanded as a sum of $2^n$ terms, each bounded by $C^n$, from the bounded-correlation property of the $B_i$. Eq. \eqref{ApEq:Bounded_correlation_tilde} can in turn be used to show that the operators $D_i$ also satisfy the bounded correlation property as
 \begin{equation}
    \label{ApEq:Bounded_correlation_D}
    \left|\operatorname{Tr}_B\lf D_{i_1} \dots D_{i_n} \rho_B \ri \right| \leq \sum_{i_1, \dots i_n  = 1}^K \left |g _{i_1} \right| \dots \left| g_{i_n} \right| \left|\operatorname{Tr}_B\lf \tilde{B}_{i_1} \dots \tilde{B}_{i_n} \rho_B \ri \right| \leq \lf 2C K\max_{i} g_i \ri^n, 
 \end{equation}
 where in the first inequality we have used the fact that $|u_{ij}| \leq 1$, as $u$ is a unitary matrix and the triangle inequality. Define the operator $Z_{\text{eff}}$ as,
\begin{equation}
     \label{ApEq:ZS_defn}
     Z_{\text{eff}} :=  \omega_0 G_\parallel + \sum_{k = 0}^K g_{k} \mean{B_k} E_k =  \sum_{k = 0}^K \lf g_{k}\mean{B_k} + \omega_0 \operatorname{Tr}(G E_k) \ri E_k
\end{equation}
In the way described by Lemma~\ref{lem:Error_detection}, we use an additional noiseless ancilla to construct a codespace based on $G_\perp$. Note that the code has the properties
 \begin{equation}
 \label{apEq:Props_CS}
     \Pi_C \mathds{1}_A \otimes G_\perp \Pi_C = G_{\text{eff}}, \quad \Pi_C \mathds{1}_A \otimes E_m \Pi_C = \xi_{m} \Pi_C , \quad \Pi_C \mathds{1}_A \otimes Z_{\text{eff}} \Pi_C = \lambda \Pi_C.
 \end{equation}
Further, let $0 < \Lambda < \infty $ be a positive number such that 
\begin{equation}
    |g_m| \|E_m\|, |\omega_0|\|G_\perp\|, \|Z_{\text{eff}}\| \leq \Lambda,
\end{equation}
which is guaranteed to exist as all the operators act on a finite-dimensional Hilbert space and as a consequence are bounded. Note that the constant $\Lambda$ has units of frequency. Under the assumption that there exists a positive number $C >0$ such that the bath correlation functions satisfy 
\begin{equation}
    \left| \operatorname{Tr} \lf B_{i_1} \dots B_{i_n} \rho_B\ri\right| \leq C^n,
\end{equation}
the Dyson series expansion for $\mathcal{E}_{t}$ converges absolutely for all $t$. This is well-known in the literature~\cite{breuer_book}, but for the sake of completeness, we sketch a proof for the above statement. Consider a Hamiltonian of the form $H = \sum_{r = 1}^{R} P_r \otimes Q_r$. Let $\mathcal{L}(\odot) = -i [H,\odot]$ be the associated Liouvillian. Consider the following chain of inequalities, 
\begin{align}
    \label{ApEq:BCFP_line_a}
    \left \| \operatorname{Tr}_{B} \lf \mathcal{L}^n \rho_S \otimes \rho_B  \ri \right\| & \leq \sum_{r_1 = 1}^{R} \dots \sum_{r_n = 1}^{R} \left\| \operatorname{Tr}_{B}[P_{r_n} \otimes Q_{r_n}, \dots,[P_{r_1}\otimes Q_{r_1}, \rho_S \otimes \rho_B ]\right\| \\
     \label{ApEq:BCFP_line_b}
    & \leq  \sum_{r_1 = 1}^{R} \dots \sum_{r_n = 1}^{R} \sum_{i = 1}^{2^n} \max_{k} \|P_k\|^n \tilde{C}^n, \\
     \label{ApEq:BCFP_line_c}
    & = \lf 2 R \tilde{C}  \max_{k} \|P_k\| \ri^n.
\end{align}
In line \eqref{ApEq:BCFP_line_a} we have used the triangle inequality, and in line \eqref{ApEq:BCFP_line_b} we have expanded the commutator and used the triangle inequality, the bounded correlation function property, and the sub-multiplicative property of operator norm. The constant $\tilde{C}$ is the constant associated with the bounded correlation of the operators $Q_i$ with respect to $\rho_B$. Note that the commutator has $2^n$ terms, with each term containing $n$ of the $P_k$ operators and an $n$-point correlation function that comes from tracing out the environment. Finally, we can use the above property to show that the Taylor expansion of the channel converges. For any $M > 0$ 
\begin{align}
    \left\|\operatorname{Tr}_B \lf \exp\lf t \mathcal{L}\ri \rho_S \otimes \rho_B \ri - \sum_{m = 0}^M \frac{\operatorname{t^{m}Tr_B(\mathcal{L}^m \rho_S \otimes \rho_B)}}{m!} \right\| &\leq  \sum_{m = M+1}^\infty \frac{t^m}{m!}\left\| \operatorname{Tr_B(\mathcal{L}^m \rho_S \otimes \rho_B}) \right\| \\
    &\leq \frac{\lf 2 R C  \max_{k} \|P_k\| t\ri^{M+1}}{(M+1)!
    } \exp\lf 2 R C  \max_{k} \|P_k\| t \ri, \\
    & \longrightarrow 0 \text{ for } M \longrightarrow \infty.
\end{align}
Since the expansion converges, we begin by perturbatively expanding the state $\rho(\delta t) = \mathcal{E}_{\delta t} \lf \rho_S \otimes \rho_B \ri $ to the second order in $\delta t$
 \begin{align}
 \label{Apeq:Infinitesimal_evol}
   \rho(\delta t) = \rho_S  &- (i \delta t) \Big( [\omega_0 \mathds{1}_A \otimes G_\perp, \rho_S] + [\mathds{1}_A\otimes Z_{\text{eff}}, \rho_S]\Big) \notag \\
   & -\frac{\delta t^2}{2} \Big( \big\{ \mathds{1}_A \otimes \lf \omega_0 G_\perp + Z_{\text{eff}} \ri^2, \rho_S \big\} - 2 \mathds{1}_A \otimes \lf \omega_0 G_\perp + Z_{\text{eff}} \ri \rho_S \mathds{1}_A \otimes \lf \omega_0 G_\perp + Z_{\text{eff}} \ri \notag \\
   & + \sum_{k = 0}^K \mean{D_k^\dagger D_k}\{\mathds{1}_A \otimes F_{k}^{\dag}F_{k} , \rho_S\} - 2\sum_{k=0}^{K} \mean{D_k^\dagger D_k} (\mathds{1}_A \otimes F_k^{\dag}) \rho_s (\mathds{1}_A \otimes F_k) \Big) \notag  \\
   & + O\lf (K+1)^{3} \Lambda^3\max\lf \left[ 2 C K \max_i g_i\right]^3,1 \ri  \ri.
 \end{align}
In deriving the above equation, we have used the property that $\operatorname{Tr(\rho_B) = 1}$, Eq. \eqref{ApEq:Bounded_correlation_D} and ignored the terms containing $\mean{D_k}$ as they are zero.  To evaluate the result of the two outcome POVM measurement $\{\Pi_C, \Pi_{\tilde{C}}\}$, first note that using $\rho_s = \Pi_C \rho_s \Pi_C$, we have 
\begin{align}
\label{apEq:Props1a}
    &\Pi_C [O,\rho_S] \Pi_C =  \Pi_C [O,\Pi_C \rho_S \Pi_C ] \Pi_C = [\Pi_C O \Pi_C, \rho_s], \\ 
    &\Pi_{\tilde{C}} [O,\rho_S] \Pi_{\tilde{C}} =  0, \\
    &\Pi_C \{O,\rho_S\} \Pi_C =  \Pi_C \{O,\Pi_C \rho_S \Pi_C \} \Pi_C = \{\Pi_C O \Pi_C, \rho_s\}, \\
    &\Pi_{\tilde{C}} \{O,\rho_S\} \Pi_{\tilde{C}} =  0,\\
     \label{apEq:Props1b}  
    &\Pi_C O^2 \Pi_C =  \Pi_C O \Pi_C O \Pi_C + \Pi_C O \Pi_{\tilde{C}} O \Pi_C,
\end{align}
 for any operator $O$ acting on the system-ancilla Hilbert space. Using Eq.~\eqref{apEq:Props_CS}, Eq.~\eqref{Apeq:Infinitesimal_evol} and Eqs.~\eqref{apEq:Props1a}-\eqref{apEq:Props1b}, we can show that
\begin{equation}
    \begin{aligned}
        \Pi_{\tilde{C}} \rho(\delta t) \Pi_{\tilde{C}} =   \delta t^2 \Big(  &\Pi_{\tilde{C}} \mathds{1}_A \otimes  (\omega_0 G_\perp + Z_{\text{eff}}) \Pi_C \rho_S \Pi_C \mathds{1}_A \otimes (\omega_0 G_\perp + Z_{\text{eff}}) \Pi_{\tilde{C}} \\&  + \sum_{k = 0}^K \mean{D_{k}^{\dag}D_{k}} \Pi_{\tilde{C}}(\mathds{1}_A \otimes F^{\dag}_k)\Pi_C \rho_S \Pi_C (\mathds{1}_A \otimes F_k )\Pi_{\tilde{C}}  \Big)  
       + O\lf (K+1)^{3} \Lambda^3\max\lf \left[ 2 C K \max_i g_i\right]^3,1 \ri  \ri. \\
    \end{aligned}
\end{equation}
Thus, the infinitesimal Kraus operators are given as 
\begin{equation}
   \label{Apeq:Infinitesimal_Krauss}
    \begin{aligned}
        N_i &= \Pi_{\tilde{C}} \lf \mathds{1}_A \otimes F_i   \ri  \Pi_C\sqrt{\langle D_i^\dagger D_i  \rangle} \delta t = \sum_{j}u^*_{ij}\lf \mathds{1}_A \otimes \lf E_j - \xi_j \ri \ri \Pi_C \sqrt{\langle D_i^\dagger D_i \rangle} \delta t, && \text{for } i \in \{0,\dots,K\}, \\
        N_i &= \Pi_{\tilde{C}}  \lf \mathds{1}_A \otimes (Z_{\text{eff}}  + \omega_0 G_\perp) \ri\Pi_C \delta t =\mathds{1}_{A}\otimes (Z_{\textrm{eff}}-\lambda)\Pi_{c}\delta t+\omega_{0}(\mathds{1}_{A}\otimes G_{\perp}-G_{\textrm{eff}})\Pi_{C}\delta t,  && \text{for } i = K+1.
    \end{aligned}
\end{equation}
For us to be able to construct a recovery operator $\mathcal{R}_{\tilde{\mathcal{C}}}$ such that $\mathcal{R}_{\tilde{\mathcal{C}}} \lf \sum_i N_i \Pi_C \rho \Pi_C N^\dagger_i  \ri \propto \Pi_C \rho \Pi_C$, we need to satisfy the Knill-Laflamme conditions~\cite{KnillLaflamme}, hence we need 
\begin{align}
        \label{eq:KL_EC_unsimp_1}
        &\Pi_C \mathds{1}_A \otimes \lf E_i  - \xi_i  \ri^\dagger \lf E_j  - \xi_j  \ri \Pi_C \propto \Pi_C, \\
        \label{eq:KL_EC_unsimp_3}
        &\Pi_C \mathds{1}_A \otimes \lf E_i  - \xi_i  \ri^\dagger \Pi_{\tilde{C}}  \mathds{1}_A \otimes \lf Z_{\text{eff}} + \omega_0 G_\perp \ri \Pi_C \propto \Pi_C \\
        \label{eq:KL_EC_unsimp_2}
        &\Pi_C \mathds{1}_A \otimes \lf Z_{\text{eff}} + \omega_0 G_\perp \ri^\dagger \Pi_{\tilde{C}}  \mathds{1}_A \otimes \lf Z_{\text{eff}} + \omega_0 G_\perp \ri \Pi_C \propto \Pi_C,
\end{align}
where we have used the fact that the $N_{j}$ for $j<K+1$ are proportional to unitary linear combinations of the $E_{k}-\xi_{k}$. Using Eq.~\eqref{apEq:Props_CS}, we see that Eq.~\eqref{eq:KL_EC_unsimp_1} can be satisfied if and only if 
\begin{equation}
    \label{Apeq:AQEC_final_cond_1}
    \Pi_C \mathds{1}_A \otimes E_iE_{j} \Pi_C \propto \Pi_C.
\end{equation}
Similarly, we can show that Eq.~\eqref{eq:KL_EC_unsimp_3} can be satisfied if and only if we additionally have
\begin{equation}
    \label{Apeq:AQEC_final_cond_2}
    \Pi_C \mathds{1}_A \otimes E_i  \Pi_{\tilde{C}}  \mathds{1}_A \otimes G_\perp \Pi_C \propto \Pi_C.
\end{equation}
Finally, Eq.~\eqref{eq:KL_EC_unsimp_2} can be satisfied if and only if the additional property
\begin{equation}
    \label{Apeq:AQEC_final_cond_3}
   \Pi_C \mathds{1}_A \otimes G_\perp \Pi_{\tilde{C}}  \mathds{1}_A \otimes G_\perp \Pi_C \propto \Pi_C,
\end{equation}
is satisfied.
\end{proof}
We conclude this section by remarking that the condition in Theorem \ref{Thm:Recovery_cond} is similar to the HNLS conditions for Markovian metrology. Thus, the noise models for which we can construct an active recovery operation is a smaller set than the ones for which we can recover Heisenberg-limited scaling using the Zeno limit.  

\section{Proof of Theorem 3}
\label{Ap:Proof_Theorem_3}
\begin{thm*}[\textbf{\ref{Thm:Error_Order_AQEC}}]
     \label{ApThm:Error_Order_AQEC_AQED}
   Consider the same setup and joint Hamiltonian as in Eq.~\eqref{Eq:Setup_Hamiltonian}. Let $U_{\text{tot}}(\delta t)$ be the unitary generated by the Hamiltonian, $U_{\textrm{tot}}(\delta t):= \exp(-i \delta t H_\text{tot}) $ and let $\mathcal{U}$ be the superoperator associated with this unitary. Suppose there exists a positive constant $C >0$, such that the $n$-point bath correlation functions satisfy 
    \begin{equation}
        \left|\operatorname{Tr}_B\lf B_{i_1} \dots B_{i_n} \rho_B \ri \right| \leq C^n.
    \end{equation}
    Define the channel $\mathcal{E}$ as
    \begin{equation}
        \mathcal{E}(\rho) := \operatorname{Tr}_B  \lf \mathds{1}_A \otimes U_{\text{tot}}(\delta t) \rho \mathds{1}_A \otimes U^\dagger_{\text{tot}}(\delta t)\ri.
    \end{equation}
  Suppose that the codespace constructed according to Lemma~\ref{lem:Error_detection} satisfies the conditions
    \begin{align}
         \label{Apeq:Rec_cond_a}
         \Pi_C &E_i E_j \Pi_C \propto \Pi_C, \\
          \label{Apeq:Rec_cond_b}
         \Pi_C &\mathds{1}_A \otimes E_{j}  \Pi_{\tilde{C}}  \mathds{1}_A \otimes G_\perp \Pi_C \propto \Pi_C, \\
          \label{Apeq:Rec_cond_c}
         \Pi_C &\mathds{1}_A \otimes G_\perp  \Pi_{\tilde{C}}  \mathds{1}_A \otimes G_\perp \Pi_C \propto \Pi_C,
    \end{align}
    with the first two proportionalities holding for all $i$ and $j$. Theorem~\ref{Thm:Recovery_cond} then implies the existence of a channel 
    \begin{align}
        \mathcal{R} \circ \mathcal{E} (\rho_S \otimes \rho_B)  &= \Pi_C \mathcal{E} (\rho_S \otimes \rho_B) \Pi_C   + \mathcal{R}_{\tilde{\mathcal{C}}} \lf\Pi_{\tilde{C}} \mathcal{U} (\rho_S \otimes \rho_B) \Pi_{\tilde{C}} \ri,
    \end{align}
   where $\mathcal{R}_{\tilde{\mathcal{C}}}$ satisfies 
   \begin{equation}
         \label{ApEQ:Errorspace_recovery}
        \mathcal{R}_{\tilde{\mathcal{C}}} \lf \Pi_{\tilde{C}} \mathcal{E}(\rho) \Pi_{\tilde{C}}\ri \propto \rho + O(\delta t^3), \quad \text{for} ~ \rho = \Pi_C \rho \Pi_C.
    \end{equation}
   Then we have 
   \begin{align}
       \mathcal{R} \circ \mathcal{E} (\rho_S \otimes \rho_B)  &=  V_{\text{exact}}(\delta t) \rho_s  V^\dagger_{\text{exact}}(\delta t) +  O(\delta t^3),
   \end{align}
   where the unitary $V_{\text{exact}}(\delta t)$ is given as 
   \begin{equation}
   \label{Apeq:exact_evo}
       V_{\text{exact}}(\delta t) := \exp\left[ -i \delta t \lf \mathds{1}_A \otimes G_{\textrm{eff}} \ri \right] \Pi_C
   \end{equation}
\end{thm*}

\begin{proof}
Since the conditions in Eqs.~\eqref{Apeq:Rec_cond_a},~\eqref{Apeq:Rec_cond_b} and~\eqref{Apeq:Rec_cond_c} are satisfied, Theorem~\ref{Thm:Recovery_cond} implies that we can construct a recovery operator, that satisfies Eq.~\eqref{EQ:Errorspace_recovery}. In addition to the properties in Eq.~\eqref{apEq:Props_CS}, the code can be seen to satisfy 
\begin{equation}
    \label{apEq:Props_CS_extended}
     \Pi_C \mathds{1}_A \otimes Z^2_{\text{eff}} \Pi_C = \Lambda \Pi_C, \quad \Pi_C \mathds{1}_A \otimes Z_{\text{eff}}  \Pi_{\tilde{C}}  \mathds{1}_A \otimes G_\perp \Pi_C = \gamma_1 \Pi_C, \quad \Pi_C \mathds{1}_A \otimes G_\perp  \Pi_{\tilde{C}}  \mathds{1}_A \otimes G_\perp \Pi_C = \gamma_2 \Pi_C,  
\end{equation}
\begin{equation}
  \label{apEq:Props_CS_extended_0}
    \Pi_C \mathds{1}_A \otimes F_i^\dagger F_j \Pi_C  = \chi_{ij}, \quad \Pi_C \mathds{1}_A \otimes  F_j \Pi_C  = \tilde{\xi}_{ij}. 
\end{equation} Furthermore, we can also show that
\begin{equation}\label{Apeq:Props_CS_extended_1}
    \Pi_C (\mathds{1}_A \otimes \omega_0 G_\perp Z_{\text{eff}}) \Pi_C = \lambda \Pi_C (\mathds{1}_A \otimes \omega_0 G_\perp) \Pi_C + \gamma_1 \Pi_C
\end{equation}Define the operator $G_{\text{eff}}$ as, 
\begin{equation}
\label{ApEq:G_eff_defn}
    G_{\text{eff}} := \Pi_C \mathds{1}_A \otimes G_\perp \Pi_C 
\end{equation}
Projecting the state in Eq.~\eqref{Apeq:Infinitesimal_evol} to the codespace and using Eqs.~\eqref{apEq:Props_CS} and \eqref{apEq:Props_CS_extended}-\eqref{ApEq:G_eff_defn}, we have 
\begin{equation}
    \begin{aligned}
        \Pi_C \rho(\delta t) \Pi_C &=  \rho_S - (i \delta t)  [\omega_0 G_{\text{eff}}, \rho_S] -\frac{\delta t^2}{2} \Bigg[ \omega^2_0 \big\{ G_{\text{eff}}^2 , \rho_S \big\} - 2 \omega^2_0  G_{\text{eff}} \rho_S  G_{\text{eff}} \Bigg] \\
        & \qquad \qquad \qquad ~~~~~~~~~~~~~ - \delta t^2 \Bigg[   \Lambda - \lambda^2 + \omega_0 \gamma_1 + \omega_0 \gamma_1^{*} + \omega_0^2 \gamma_2 +  \sum_{k = 0}^K \mean{D_k^\dagger D_k} (\chi_{kk} - |\tilde{\xi}_{k}|^2)  \Bigg]\rho_S.
    \end{aligned}
\end{equation}
The Kraus operators defined in Eq.~\eqref{Apeq:Infinitesimal_Krauss} have the normalization
\begin{equation}
    \sum_{k = 0}^{K+1} N_k^\dagger N_k = (\delta t^2)\lf \Lambda - \lambda^2 + \omega_0 \gamma_1 + \omega_0 \gamma_1^{*} + \omega_0^2 \gamma_2 +  \sum_{k = 0}^K \mean{D_k^\dagger D_k} (\chi_{kk} - |\tilde{\xi}_{k}|^2) \ri \mathds{1} + O(\delta t^3).
\end{equation}
Thus 
\begin{equation}
    \Pi_C \rho(\delta t) \Pi_C + \mathcal{R}_{\tilde{\mathcal{C}}}(\Pi_{\tilde{C}} \rho(\delta t) \Pi_{\tilde{C}}) = \Pi_C \exp \lf -i \delta t G_{\text{eff}} \ri \rho_S \exp \lf i \delta t G_{\text{eff}} \ri \Pi_C + O(\delta t^3).
\end{equation}
Thus, the evolution matches a unitary evolution up to $O(\delta t^3)$. 
\end{proof}

\section{Conditions for reset recovery}
\label{Ap:Reset_conds}
In this section, we will derive the conditions for which the projection of the state onto the error subspace is independent of the parameter $\omega_0$ up to the leading order in $\delta t$. This would imply that for channels satisfying the said conditions, resetting the state to the initial state is the most reasonable recovery operator. 

Suppose that the initial state is evolved under an AQED code for a time $s$ before the state is projected onto the error subspace for the first time. The effective state in the error subspace would be 
\begin{equation}
    \mathcal{E}_{\tilde{\mathcal{C}}}(\rho) = \sum_{k = 0}^{K+1} K_k e^{-i \omega_0 G_{\text{eff}} s} \rho_0 e^{+i \omega_0 G_{\text{eff}} s}  K_k ^\dagger,
\end{equation}
 where, the Kraus operators are given in Eq.~\eqref{Apeq:Infinitesimal_Krauss}. If the state is independent of $\omega_0$, then we have
\begin{equation}
    \frac{\partial \mathcal{E}_{\tilde{\mathcal{C}}}(\rho)}{\partial \omega_0} = 0.
\end{equation}
 Performing the derivative, and noting that 
 \begin{equation}
      \frac{\partial K_k}{\partial \omega_0} = \delta_{k (K+1)} \Pi_{\tilde{C}} \lf \mathds{1}_A \otimes G \ri \Pi_{C}\delta t, 
 \end{equation}
 we have 
\begin{equation}
    \lf \mathds{1}_A \otimes G \ri e^{-i G_{\text{eff}} s} \rho_0 e^{i G_{\text{eff}} s} K^\dagger_{k+1} +  K_{k+1} e^{-i G_{\text{eff}} s} \rho_0 e^{i G_{\text{eff}} s} \lf \mathds{1}_A \otimes G \ri  + i s\sum_k K_k e^{-iG_{\text{eff}} s} [G_{\text{eff}}, \rho_0] e^{-i G_{\text{eff}} s} K^\dagger_k = 0.
\end{equation}
Spontaneous decay satisfies the above conditions for all times $s$. Consider the Hamiltonian,
\begin{equation}
    H = \omega_0 \sigma_+\sigma_- + \sum_k \omega_k b_k ^\dagger b_k + \sum_k \sigma_+ \otimes g_k b_k  + \sigma_- \otimes g^*_k b^\dagger_k. 
\end{equation}
Comparing with the general equation form in Eq. \eqref{ApEq:Setup_Hamiltonian}, we see that, 
\begin{equation}
    G = Z, E_1 = \sigma_+, E_2 = \sigma_-, B_1 = \sum_{k} g_k b_k \text{ and } B_2 = \sum_k g^*_k b_k^\dagger,  
\end{equation}
where $b_k(b_k^\dagger)$ are the lowering (raising) operator on environment mode $k$. The codespace and the error subspace are defined as, 
\begin{equation}
    \Pi_C = \ketbra{00}{00} + \ketbra{11}{11}, \quad \Pi_{\tilde{C}} = \ketbra{01}{01} + \ketbra{10}{10}.
\end{equation}
Furthermore, for an environment state that is initialized in the vacuum it is easy to see that $\mean{B_1} = \mean{B_2} = \mean{B_1^\dagger B_2} = \mean{B_1^\dagger B_1} = 0$ and $\mean{B_2^\dagger B_2} = \sum_k |g_k|^2$. 
Thus, we can apply the formalism adopted in the above section directly to show that the Kraus operators after the projection onto the error subspace are given by
\begin{equation}
    G_{\parallel} = 0, \quad  G_{\perp} = Z,  \quad G_{\text{eff}} = \mathds{1} _A \otimes Z, \quad Z_{\text{eff}} = 0,
\end{equation}
\begin{equation}
   \rho_0 = \frac{1}{2}\lf \ketbra{00}{00} + \ketbra{00}{11} + \ketbra{11}{00} +\ketbra{11}{11}\ri,
\end{equation}
\begin{equation}
   K_1 = 0, \quad  K_2 = \delta t \sigma_- \lf \sum_k |g_k|^2 \ri, \quad  K_{i \neq 1,2} = 0.
\end{equation}
 It is easy to see that $\sigma_-[Z,U(s) \rho_0 U(s)^\dagger]\sigma_+ = 0$. Thus, the error space in spontaneous decay has no information about $\omega_0$ and hence the reset is the only recovery we can apply.  

\section{Leading order error analysis for AQED}
\label{Ap:Error_AQED}
Let $H_{\text{tot}}$ in Eq.~\eqref{ApEq:Setup_Hamiltonian} denote the total Hamiltonian and $H_S,, H_B, H_{SB}$ denote the system, environment, and the interaction parts, respectively. For simplicity, we consider that $H_{\text{tot}}$ is a bounded linear operator.  Let $C$ and $\tilde{C}$ denote the codespace and the error space, respectively. The codespace satisfies 
\begin{equation}
    \Pi_C E_i \Pi_C  = \xi_i \mathds{1}.
\end{equation}
We start by evaluating
\begin{equation}
    V(N,\tau) := \lf \Pi_C U(\tau) \Pi_C\ri^N.
\end{equation}
Define 
\begin{equation}
    V(\tau) = \Pi_C U(\tau) \Pi_C, \qquad W(\tau) = \Pi_{\tc} U(\tau) \Pi_C.
\end{equation}
Taking derivative of $V(\tau)$ with $\tau$ we get 
\begin{equation}
\label{ApEq:Derivative_V}
    \frac{\partial V}{\partial \tau} = - i ~\Pi_C H_{\text{tot}} U(\tau) \Pi_C = - i ~ \Pi_C H_{\text{tot}} \lf \Pi_C + \Pi_{\tc} \ri U(\tau)  \Pi_C = -i H_{\text{tot}}^{CC} V + -i H_{\text{tot}}^{C\tc} W,
\end{equation} 
where $H_{\text{tot}}^{CC} = \Pi_C H_{\text{tot}} \Pi_C$ and $H_{\text{tot}}^{C\tc} = \Pi_C H_{\text{tot}} \Pi_C$. Similarly taking the derivative of $W$ yields
\begin{equation}
\label{ApEq:Derivative_W}
    \frac{\partial W}{\partial \tau}  = -i H_{\text{tot}}^{\tc C} V -i H_{\text{tot}}^{ \tc\tc} W.
\end{equation} 
Using the boundary condition $W(0) = 0$, formally integrating Eq.~\eqref{ApEq:Derivative_W} and substituting into Eq.~\eqref{ApEq:Derivative_V} we get 
\begin{equation}
   \label{ApEq:Integro_Diff}
   \frac{\partial V}{\partial \tau} =   -i H_{\text{tot}}^{CC} V - \int_0^\tau ds ~ H_{\text{tot}}^{C \tc} \exp\left[ -i (\tau - s) H_{\text{tot}}^{\tc \tc} \right] H_{\text{tot}}^{\tc C} V(s).
\end{equation}
Expanding the exponential operator, using $\|V(\tau)\| \leq \| U(\tau)\| = 1$, the $n^{th}$ order in $(\tau -s)$ term has norm 
\begin{equation}
     \label{ApEq:n_th_order_term_analysis}
    \left \| \int_0^\tau ds~ H_{\text{tot}}^{C \tc} \lf-i (\tau - s) H_{\text{tot}}^{\tc \tc} \ri^n H_{\text{tot}}^{\tc C} V(s) \right\| \leq \| H_{\text{tot}}\|^{n+2} \frac{\tau^{n+1}}{n+1}  = O\lf \|H_{\text{tot}}\| (\|H_{\text{tot}}\|\tau)^{n+1}\ri.
\end{equation}
Thus for $\tau \ll \|H_{\text{tot}}\|^{-1}$, perturbatively expanding Eq.~\eqref{ApEq:Integro_Diff} to the first order in $\tau$, we get 
\begin{equation}
    \label{ApEq:Perturbative_diffeq_error_AD}
     \frac{\partial V}{\partial \tau} =   -i H_{\text{tot}}^{CC} V - \int_0^\tau ds ~ H_{\text{tot}}^{C \tc}  H_{\text{tot}}^{\tc C} V(s) + O(\|H_{tot}\|^3\tau^2).
\end{equation}
Taking derivatives with respect to $\tau$ on both sides, we get 
\begin{equation}
   \label{ApEq:V_second_order_equation}
    \frac{\partial^2 V}{\partial \tau^2} =   -i H_{\text{tot}}^{CC} \frac{\partial V}{\partial \tau} -  H_{\text{tot}}^{C \tc}  H_{\text{tot}}^{\tc C} V(\tau) + O(\|H_{tot}\|^3\tau).
\end{equation}
Consider the operator $\tilde{V}$, given as
\begin{equation}
    \label{ApEq:V_tilde}
    \tilde{V}(\tau) = \exp\lf -i \tau H_{\text{tot}}^{CC}  - \frac{\tau^2}{2} H_{\text{tot}}^{C \tc}  H_{\text{tot}}^{\tc C}  \ri \Pi_C.
\end{equation}
We claim that the operator $\tilde{V}$ is the solution to the differential equation in Eq.~\eqref{ApEq:V_second_order_equation} to the order of $O(\|H_\text{tot}\|^3 \tau^3)$. To do this, we first need the following lemma. 

\begin{lemma}
\label{Aplem:Diff_approx}
    Let $X(t)$ be a matrix valued function that is at least twice differentiable. Suppose we have a second order matrix differential equation, 
       \begin{equation}
           \frac{\partial^2 X}{\partial t^2}  + A \frac{\partial X}{\partial t} + B X= E(t), 
       \end{equation}
       where $A$, $B$ , and $E$ are bounded linear operators and $\|E(t)\| = O(t^p)$ as $t$ goes to $0$, for some $p > 0$. If X also satisfies the boundary conditions
       \begin{equation}
           X(0) = \frac{\partial X}{\partial t} (0) = 0.
       \end{equation}
       Then 
       \begin{equation}
            \|X(t)\| = O(t^{p+2}),
       \end{equation}
       as $t$ goes to $0$.
\end{lemma}
\begin{proof} The proof follows by bounding the solution derived using the Duhamel principle. Define the operators $Y(t)$, $M(t)$ and $D(t)$ as  
\begin{equation}
    Y(t) = \begin{bmatrix}
        X(t) \\ X'(t)
    \end{bmatrix},\quad M(t) = \begin{bmatrix}
        0 & \mathds{1} \\ -B & -A
    \end{bmatrix},\quad D(t) = \begin{bmatrix}
    0 \\ E(t)
\end{bmatrix}. 
\end{equation}
Then, the above differential equation can be written as, 
\begin{equation}
    Y'(t) = M Y(t) + D(t).
\end{equation}
Note that $M$ is bounded since $A$ and $B$ themselves are bounded. Using Duhamel's principle (see Section 2.7 of Ref.~\cite{hochstadt2014differential}), we can write the solution to $Y(t)$ as 
\begin{equation}
    Y(t) = e^{M t} Y(0) + \int_0^t ds ~ e^{M(t - s)} D(s) =  \int_0^t ds ~ e^{M(t - s)} D(s),
\end{equation}
where the second equality follows from the boundary conditions. Thus
\begin{align}
    Y(t) &= \int_0^t ds~ D(s) + \int_0^t ds~ M (t - s) D(s) + \sum_{n= 2}^\infty \int_0^t ds~\frac{\lf M(t  -s) \ri^n}{n!} D(s), \\
    &=\begin{bmatrix}
        X(t) \\ X'(t)
    \end{bmatrix} =  \int_0^t ds~ \begin{bmatrix}
        (t - s) E(s) \\ E(s) - (t - s) A E(s) 
    \end{bmatrix} + O \lf \frac{\|M\|^2 t^{p+3}}{2 (p+1)} e^{\|M\|t}\ri.
\end{align}
The second equality above is derived by upper-bounding the third term in the first equality as
\begin{align}
    \lim_{t \to 0}\left \| \sum_{n= 2}^\infty \int_0^t ds~\frac{\lf M(t-s) \ri^n}{n!} D(s) \right\| &\leq  \lim_{t \to 0} \sum_{n= 2}^\infty \int_0^t ds~\frac{\lf \|M\|t) \ri^n}{n!} \|D(s)\|,  \\ 
    &\leq \lim_{t \to 0} \sum_{n= 2}^\infty \int_0^t ds~\frac{\lf \|M\|t) \ri^n}{n!} s^p, \\
    &\leq \lim_{t \to 0} \frac{t^{p+3} \|M\|^2}{2(p+1)} \sum_{n = 0}^\infty \frac{\lf \|M\|t) \ri^n}{n!} \frac{n! {2!}}{(n+2)!}, \\
    &<  \lim_{t \to 0}\frac{t^{p+3} \|M\|^2}{2(p+1)} e^{\|M\|t}.
\end{align}
Thus 
\begin{equation}
    \| X(t) \| \leq \int_0^t ds~ (t - s)\|E(s)\| <  t \int_0^t ds ~ s^p =  O \lf t^{p+2}\ri
\end{equation}
as $t$ goes to $0$.
\end{proof}

To show that the operator in Eq.~\eqref{ApEq:V_tilde} is the solution to Eq.~\eqref{ApEq:V_second_order_equation}, we start by noting that, by Taylor expanding $\tilde{V}(\tau)$ on the order of $O(\tau^2)$ and substituting in Eq.~\eqref{ApEq:V_second_order_equation} we obtain
\begin{equation}
    \frac{\partial^2 \tilde{V}}{\partial \tau^2} =   -i H_{\text{tot}}^{CC} \frac{\partial \tilde{V}}{\partial \tau} -  H_{\text{tot}}^{C \tc}  H_{\text{tot}}^{\tc C} \tilde{V}(\tau) + O(\|H_{tot}\|^3\tau).
\end{equation}
Thus, we have 
\begin{equation}
    \lf \frac{\partial^2}{\partial \tau^2} + i H_{\text{tot}}^{C \tilde{C}} \frac{\partial}{\partial \tau} + H_{\text{tot}}^{C \tilde{C}}H_{\text{tot}}^{ \tilde{C}C}\ri \lf \tilde{V} - V \ri =  O\lf \|H_{tot}\|^3\tau\ri.
\end{equation}
Now, consider the boundary conditions on the operators $V$ and $\tilde{V}$. Clearly 
\begin{equation}
    V(0) = \tilde{V}(0) = \Pi_C. 
\end{equation}
Using Duhamel's formula for the derivative of a matrix exponential we can also show that 
\begin{equation}
    \tilde{V}'(0) = V(0) = -i H_{\text{tot}}^{CC} \Pi_C.
\end{equation}
Thus $V - \tilde{V}$ satisfies the boundary conditions for applying Lemma \ref{Aplem:Diff_approx}.  Applying the lemma, we have 
\begin{equation}
    V(\tau) = \tilde{V}(\tau) + O( \|H_{tot}\|^3\tau^3).
\end{equation}
Thus, we have 
\begin{equation}
    \label{ApEq:V_AQED_perturb}
    V(\tau) = \exp\lf -i \tau H_{\text{tot}}^{CC}  - \frac{\tau^2}{2} H_{\text{tot}}^{C \tc}  H_{\text{tot}}^{\tc C}  \ri \Pi_C + O(\|H_\text{tot}\|^3 \tau^3).
\end{equation}
Finally, note that $V(N,\tau) = V(\tau)^N$, thus 
\begin{equation}
    V(N,\tau) = \exp\lf -iN \tau H_{\text{tot}}^{CC}  - \frac{N \tau^2}{2} H_{\text{tot}}^{C \tc}  H_{\text{tot}}^{\tc C}  \ri \Pi_C + O(N\|H_\text{tot}\|^3 \tau^3).
\end{equation}
Substituting Eq.~\eqref{ApEq:V_AQED_perturb} in Eq.~\eqref{ApEq:Derivative_W} and expanding first order in $\tau$ using arguments similar to the ones used in Eq.~\eqref{ApEq:n_th_order_term_analysis} we obtain
\begin{equation}
    \label{ApEq:Perturb_W}
    W(\tau) = -i \int_0^\tau ds ~ \exp\left[-i(\tau - s) H_{\text{tot}}^{\tilde{C}\tilde{C}} \right] H^{\tilde{C}C}_{\text{tot}} V(s) = -i H^{\tilde{C}C}_{\text{tot}} \tau + O \lf \|H_{\text{tot}}\|^2 \tau^2\ri.
\end{equation}
Recalling the error detection scheme, if the state is in the codespace, we do nothing and if it is in the error space, we reset. Thus, after a single round of error correction, we get the channel $\mathcal{R}\circ \mathcal{E}$ acting on an initial state in the codespace $\rho_0$ 
\begin{equation}
   \mathcal{R}\circ \mathcal{E} (\rho_0) =   V(\tau) \rho_0 V^\dagger(\tau) + \operatorname{Tr} \lf  W^\dagger(\tau) W(\tau) \rho_0 \ri \rho_0.
\end{equation}
We will now prove using induction that, 
\begin{equation}
     \label{ApEq:Final_error_AQED}
     \lf \mathcal{R}\circ \mathcal{E} \ri^N \rho_0 =  \tilde{V}^{N}(\tau) \rho_0 \tilde{V}^{\dagger N}(\tau)  + \tau^2 \left[ \sum_{k = 0}^{N - 1} \operatorname{Tr} \lf \tilde{V}^{\dagger N - 1 - k}(\tau) H_{\text{tot}}^{C \tilde{C}} H_{\text{tot}}^{\tilde{C}C} \tilde{V}^{ N - 1 - k}(\tau) \rho_0\ri \tilde{V}^{k}(\tau) \rho_0 \tilde{V}^{\dagger k}(\tau) \right] + O(N \|H_{\text{tot}}\|^3 \tau^3).
\end{equation}
 We start by bounding the error on approximating $\mathcal{R}\circ\mathcal{E}(\rho_0)$. Define the operators $C(\tau)$ and $D(\tau)$ as , 
\begin{align}
    V(\tau) &= \tilde{V}(\tau) + C(\tau), \\
    W(\tau) &= -i H^{\tilde{C}C}_{\text{tot}} \tau + D(\tau). 
\end{align}
Then, 
\begin{align}
    \label{ApEq:Bounding_line_a}
    \left\|\mathcal{R}\circ \mathcal{E} (\rho_0) -  \tilde{V}(\tau) \rho_0 \tilde{V}^\dagger(\tau) - \tau^2 \operatorname{Tr} \lf  H^{C\tilde{C}}_{\text{tot}} H^{\tilde{C}C}_{\text{tot}} \rho_0 \ri \rho_0 \right\| & \leq \left\|\tilde{V}(\tau) \rho_0 C(\tau) + C(\tau)\rho_0\tilde{V}(\tau) \right\|+ \left\|C(\tau) \rho_0 C(\tau) \right\|  \notag \\
    &~~~~+  \left\|-i\tau \operatorname{Tr} \lf D^\dagger(\tau)H^{\tilde{C}C}_{\text{tot}} \rho_0 \ri \rho_0 + i \operatorname{Tr} \lf  H^{C\tilde{C}}_{\text{tot}} D(\tau) \rho_0 \ri \rho_0 \right\| \notag \\
    &~~~~ + \left\|\operatorname{Tr} \lf D^\dagger(\tau)D(\tau) \rho_0 \ri \rho_0  \right\|,   \\
    \label{ApEq:Bounding_line_b}
    & \leq 2\|\tilde{V}(\tau)\| \| C(\tau )\| \|\rho_0\| + 2 \tau \left| \operatorname{Tr} \lf  H^{C\tilde{C}}_{\text{tot}} D(\tau) \rho_0 \ri \right| \|\rho_0\| \notag \\
    &~~~~ + \|C(\tau)\|^2 \|\rho_0\|  + 2 \left| \operatorname{Tr} \lf D^\dagger(\tau) D(\tau) \rho_0 \ri \right| \|\rho_0\|, \\
    \label{ApEq:Bounding_line_c}
    &\leq \lf 2 \| C(\tau )\| + 2 \tau \| H^{C\tilde{C}}_{\text{tot}} \| \| D(\tau) \| + \|C(\tau)\|^2  + \|D(\tau)\|^2 \ri \operatorname{Tr}(\rho_0) , \\
    \label{ApEq:Bounding_line_d}
    &\leq  \tau^3 \| H_{\text{tot}} \|^3
      \lf 4 + \tau^3 \| H_{\text{tot}} \|^3  + \tau \| H_{\text{tot}} \|\ri \operatorname{Tr}(\rho_0), \\
      \label{ApEq:Bounding_line_e}
    &\leq  O \lf \tau^3 \| H_{\text{tot}} \|^3 \ri, 
 \end{align}
 In the above chain of equations, we have used triangle inequality in line \eqref{ApEq:Bounding_line_a} and submultiplicative property of operator norm in line \eqref{ApEq:Bounding_line_b}. In line \eqref{ApEq:Bounding_line_c} we have used Holder inequality, which states that $\left|\operatorname{Tr\lf A \rho \ri}\right| \leq \|A\|_{\infty} \|\rho\|_{1}$. We arrive at the equation in line \eqref{ApEq:Bounding_line_c} by further using the fact that $\|\rho_0\| \leq \operatorname{Tr}(\rho_0) =  1$ and $\rho \geq 0$.  The equation \eqref{ApEq:Final_error_AQED} is clearly valid for $N = 1$. Assume it is valid for some $k \in \mathds{N}$. Define $\rho_k$ as 
 \begin{equation}
     \rho_k  = \tilde{V}^{k}(\tau) \rho_0 \tilde{V}^{\dagger k}(\tau)  + \tau^2 \left[ \sum_{m = 0}^{k - 1} \operatorname{Tr} \lf \tilde{V}^{\dagger k - 1 - m}(\tau) H_{\text{tot}}^{C \tilde{C}} H_{\text{tot}}^{\tilde{C}C} \tilde{V}^{ k - 1 - m}(\tau) \rho_0\ri \tilde{V}^{m}(\tau) \rho_0 \tilde{V}^{\dagger m}(\tau) \right].
 \end{equation}
 Consider $\lf \mathcal{R}\circ \mathcal{E} \ri^{k+1} \rho_0$. Using 
\begin{equation}
    \left \|\tau^2 \operatorname{Tr} \lf \tau^2 \left[ \sum_{k = 0}^{N - 1} \operatorname{Tr} \lf \tilde{V}^{\dagger N - 1 - k}(\tau) H_{\text{tot}}^{C \tilde{C}} H_{\text{tot}}^{\tilde{C}C} \tilde{V}^{ N - 1 - k}(\tau) \rho_0\ri \tilde{V}^{k}(\tau) \rho_0 \tilde{V}^{\dagger k}(\tau) \right] H_{\text{tot}}^{C \tilde{C}} H_{\text{tot}}^{\tilde{C}C} \ri \right \| = O(N \|H_{\text{tot}}\|^4\tau^4). 
\end{equation}
We can show that 
\begin{align}
   \lf \mathcal{R}\circ \mathcal{E} \ri^{k+1} \rho_0  &= \lf \mathcal{R}\circ \mathcal{E} \ri\rho_k + O(k\tau^3 \| H_{\text{tot}}^3 \| ), \\
   &= \tilde{V}(\tau) \lf\rho_k \ri \tilde{V}^\dagger(\tau) + \tau^2 \operatorname{Tr} \lf H_{\text{tot}}^{C \tilde{C}} H_{\text{tot}}^{\tilde{C}C}  \tilde{V}^{k}(\tau) \rho_0 \tilde{V}^{\dagger k}(\tau) \ri \rho_0 \\
   &+ O\lf k \tau^4 \| H_{\text{tot}}\|^4\ri + O \lf \tau^3 \| H _{\text{tot}}\|^3  \ri + O\lf k\tau^3 \| H_{\text{tot}}\|^3 \ri,  \\
   &= \tilde{V}^{k+1}(\tau) \rho_0 \tilde{V}^{\dagger k+1}(\tau)  + \tau^2 \left[ \sum_{m = 0}^{k} \operatorname{Tr} \lf \tilde{V}^{\dagger k - m}(\tau) H_{\text{tot}}^{C \tilde{C}} H_{\text{tot}}^{\tilde{C}C} \tilde{V}^{k - m}(\tau) \rho_0\ri \tilde{V}^{m}(\tau) \rho_0 \tilde{V}^{\dagger m}(\tau) \right] + O((k+1)\|H_{\text{tot}}\|^3 \tau^3)
\end{align}

\section{Leading order error analysis for Dynamical Decoupling}
\label{ap:Leading_order_DD}
Let $H_{\text{tot}}$ in Eq.~\eqref{ApEq:Setup_Hamiltonian} denote the total Hamiltonian. As before, let us assume that $H_{\text{tot}}$ is bounded and  $\|H_{\text{tot}}\| < \infty$. Let $C$ and $\tilde{C}$ denote the codespace and the error space, respectively. The codespace satisfies 
\begin{equation}
    \Pi_C E_i \Pi_C  = \xi_i \mathds{1}.
\end{equation}
 The unitary pulse $U_p$ is given as 
\begin{equation}
    U_p = \Pi_C + e^{i \phi} \Pi_{\tc}.
\end{equation}
We want to evaluate 
\begin{equation}
    V(N,\tau) := \lf U_p U(\tau) \ri^N = U_p \left[   U(\tau)U_p \right]^{{N-1}} U(\tau). 
\end{equation}
Inserting $ U_p^{j} (U^\dagger_p)^j$ after the $j^{th}$ pulse, beginning with $m=0$, we can re-write the above expression as 
\begin{equation}
      V(N,\tau) = U_p^{N} \prod_{m = 0}^{N - 1} \lf U_p^{\dagger m} U(\tau) U_p^m\ri = U_p^N \prod_{m = 0}^{N - 1} \exp \lf-i \tau  U_p^{\dagger m} H_{\textrm{tot}} U_p^m\ri.
\end{equation}
Expanding the pulse unitary in terms of its eigenprojectors, we can write 
\begin{equation}
      U_p^{\dagger m} U(\tau) U_p^m = \exp{\lf-i \tau \lf H_0 + A_m \ri \ri},
\end{equation}
where 
\begin{align}
    H_0 &=\Pi_C H_{\text{tot}} \Pi_C + \Pi_{\tc} H_{\text{tot}} \Pi_{\tc} = H_{\text{tot}}^{C C} + H_{\text{tot}}^{\tc\tc }, \\ 
    A_m &=   e^{i \phi m} \Pi_C H_{\text{tot}}\Pi_{\tc} + e^{-i \phi m} \Pi_{\tc} H_{\text{tot}}\Pi_C=  e^{i \phi m} H_{\text{tot}}^{C\tc} +  e^{-i \phi m} H_{\text{tot}}^{\tc C} . 
\end{align}
Thus, we have, 
\begin{equation}
  V(N,\tau)  = U_p^N \prod_{m = 0}^{N} \exp \lf -i \tau \lf H_0 + A_m \ri \ri.    
\end{equation}

We will evaluate the above using a perturbative version of the Baker-Campbell-Hausdorff (BCH) expansion. We will keep the leading order term and the subleading term to quantify the error. To bound the error rigorously, we will only focus on the case $\phi = \pi$. For $\phi= \pi$, we see that the operators $A_m$ have an additional structure 
\begin{equation}
    A_m = (-1)^m A_0.
\end{equation}
Thus, the operator $V(N,\tau)$ for even $N$, is reduced to 
\begin{equation}
\label{apeq:DD_V_simplified}
    V(N,\tau) = \left[ \exp(-i \tau \lf H_0 - A_0\ri) \exp(-i \tau \lf H_0 + A_0\ri) \right]^{N/2}.
\end{equation}
A rigorous bound on the above expression can be obtained by expressing the BCH terms in terms of a Magnus expansion, as described in Ref.~\cite{blanes_magnus_2009} (Sec.~2.8). The associated truncation error is then bounded using the Magnus-expansion error bounds of Ref.~\cite{sharma_hamiltonian_2024}. For the sake of completeness, we prove the following lemma, 
\begin{lemma}
    Suppose $Y_1$ and $Y_2$ are bounded self-adjoint operators in the Hilbert space $\mathcal{H}$. Let $\epsilon > 0$ and $\Lambda > 1$ be such that 
    \begin{equation}
      \label{ApEq:Lemma_cond}
        \epsilon \lf \max_{j} \| Y_j \| \ri \leq 1 - \frac{1}{\Lambda}.
    \end{equation}
    Then     \begin{equation}
          \exp \lf -i \epsilon Y_2 \ri\exp \lf -i \epsilon Y_1 \ri = \exp\lf \sum_{p = 1}^{M}\Omega_p\ri + O\lf\epsilon^{M+1}\lf\max_{j}\|Y_j\|\ri^{M+1}\ri,
     \end{equation}
     where $\Omega_p$ are the BCH terms of order $\epsilon^p$ and the notations match Ref. \cite{blanes_magnus_2009}.
\end{lemma}
\begin{proof} 
 Define the time-dependent Hamiltonian 
\begin{equation}
    \label{apeq:magnus_bch}
    H= \left\{ \begin{matrix}
        \epsilon Y_2  & 0 \leq t < 1 \\
        \epsilon Y_1  & 1 \leq t < 2.
    \end{matrix} \right.
\end{equation}
From Ref.~\cite{blanes_magnus_2009} (section 2.8), we see that the $p^{th}$ order BCH expansion coincides with $p^{th}$ order Magnus expansion of the time-dependent Hamiltonian in Eq.~\eqref{apeq:magnus_bch}. Using the Magnus expansion, we can write 
\begin{equation}
     \exp \lf -i \epsilon Y_2 \ri\exp \lf -i \epsilon Y_1 \ri = \mathcal{T}_{\longleftarrow}\exp \lf -i \int_0^2 ds H(s) \ri = \exp\lf  \sum_{p = 1}^{\infty} \Omega_p \ri,
\end{equation}
where $\mathcal{T}_{\longleftarrow}$ is the time ordering operator. The Magnus expansion can be shown to converge as the condition in Eq. \eqref{ApEq:Lemma_cond} implies that $\epsilon \lf \sum_{j = 1}^2 \| Y_j \| \ri \leq \pi$~\cite{blanes_magnus_2009}. Then, using the bound proved in Ref.~\cite{sharma_hamiltonian_2024} (App.~B), using Duhalmel's principle, we see that 
\begin{align}
    \label{ApEq:MB_line_a}
    \left\|\exp\lf  \sum_{p = 1}^{\infty} \Omega_p \ri - \exp\lf  \sum_{p = 1}^{M} \Omega_p \ri\right\| &\leq \left\| \sum_{p = M+1}^{\infty} \Omega_p  \right\|, \\
    \label{ApEq:MB_line_b}
    &\leq \sum_{p = M+1}^{\infty} \|\Omega_p\|, \\
    \label{ApEq:MB_line_c}
    &\leq \sum_{p = M+1}^{\infty} \lf \epsilon\lf\max_{j}\|Y_j\|\ri\ri^p, \\
    \label{ApEq:MB_line_d}
    &\leq \Lambda \lf \epsilon^{M+1}\lf\max_{j}\|Y_j\|\ri^{M+1} \ri,
\end{align}
where in line \eqref{ApEq:MB_line_a}, we have used the upper bound proved in Ref.~\cite{sharma_hamiltonian_2024}, in line \eqref{ApEq:MB_line_b} we have used the triangle inequality, in line \eqref{ApEq:MB_line_c} we have bounded the $p^{th}$ order Magnus term and in line \eqref{ApEq:MB_line_d} we have used the closed form formula for the sum of an infinite geometric series and Eq. \eqref{ApEq:Lemma_cond}.
\end{proof}

Noting that $\|H_0\|, \|A_0\| \leq \|H_{\text{tot}}\|$, we can apply the above lemma to Eq.~\eqref{apeq:DD_V_simplified}
\begin{equation}
    V(N, \tau) = [\exp\lf-2i \tau H_0 - \tau^2 [H_0, A_0] \ri + O\lf\|H_{\text{tot}}\|^3\tau^3\ri]^{N/2}.
\end{equation}
Expanding the above, we get.
\begin{equation}
\label{ApEq:DD_full_eq}
    V(N,\tau) = \exp \left[ -i \tau \lf N H_0  \ri  - \frac{N \tau^2}{2} \lf \left[H_0,  H_{\text{tot}}^{C \tc} + H_{\text{tot}}^{\tc C} \right] \ri\right] + O\lf \frac{N}{2} \|H_{\text{tot}}\|^3 \tau^3 \ri.
\end{equation}

\subsection{Error at the level of Fisher information}
Ultimately, we are interested in calculating the scaling at the level of quantum Fisher information. We will do this using perturbation theory. We start by re-parameterizing  the operator in Eq.~\eqref{ApEq:DD_full_eq}.  Substituting $t = N \tau$ and allowing a slight abuse of notation for clarity, we get 
\begin{equation}
    V(t,\tau) = \exp \left( -i t H_0   - \frac{t\tau}{2}  \left[H_0,  H_{\text{tot}}^{C \tc} + H_{\text{tot}}^{\tc C} \right] \right) + O\lf \frac{t}{2} \|H_{\text{tot}}\|^3 \tau^2 \ri.
\end{equation}
Taking partial derivative with respect to $t$ 
\begin{equation}
    \frac{\partial V(t,\tau)}{\partial t} = \lf- i H_0 - \frac{\tau}{2} \left[H_0,  H_{\text{tot}}^{C \tc} + H_{\text{tot}}^{\tc C} \right]  \ri V(t,\tau) + O \lf \|H_{\text{tot}}\|^3 \tau^2\ri.
\end{equation}
Going to the interaction picture with respect to $H_0$, we can show that 
\begin{equation}
    V(t,\tau) = U_0(t)  ~ T_{\longleftarrow} \exp\lf - \frac{\tau}{2} \int_0^t ds U_0^\dagger(s) \left[H_0,  H_{\text{tot}}^{C \tc} + H_{\text{tot}}^{\tc C} \right] U_0(s) \ri + O \lf t\|H_{\text{tot}}\|^3 \tau^2\ri,
\end{equation}
where $U_0(t) = \exp \lf  -i H_0 t \ri $. Expanding the interaction unitary using Dyson series and using the fact that $\Pi_C \ket{\psi_0} = \ket{\psi_0}$ and $\exp \lf  -i H_{0} t \ri \Pi_C = \exp \lf  -i H^{CC}_{\text{tot}} t \ri \Pi_C$, we get 
\begin{align}
     \label{ApEq:V_tau_expansion}
     V(t,\tau) \ketbra{\psi_0}{\psi_0}  \otimes \rho_B V^\dagger(t,\tau)  &=  \exp \lf  -i H^{CC}_{\text{tot}} t \ri  \ketbra{\psi_0}{\psi_0} \otimes \rho_B \exp \lf  +i H^{CC}_{\text{tot}} t \ri  \notag \\
     &~~~  + \sum_{n = 1}^{\infty}  \lf \frac{-\tau}{2} \ri^n \int_0^t ds_1 \dots \int_0^{s_{n - 1}} ds_n ~ U_0(t) \mathcal{L}_{\text{int}} \lf s_1,\dots \lf \mathcal{L}_{\text{int}} \lf s_n,  \ketbra{\psi_0}{\psi_0}  \otimes \rho_B \ri \ri \ri U_0^\dagger(t) \\
     &~~~+O \lf t\|H_{\text{tot}}\|^3 \tau^2\ri,
\end{align}
where  the interaction Liouvillian is given by 
\begin{equation}
    \mathcal{L}_{\text{int}} \lf s, \rho \ri = \left[ U_0^\dagger (s)\left[H_0, H_{\text{tot}}^{C \tc} + H_{\text{tot}}^{\tc C} \right] U_0(s), \rho \right].
\end{equation}
Denote the term with the interaction Liouvillian as $R$
\begin{equation}
    R := \sum_{n = 1}^{\infty}  \lf \frac{-\tau}{2} \ri^n \int_0^t ds_1 \dots \int_0^{s_{n - 1}} ds_n ~ U_0(t) \mathcal{L}_{\text{int}} \lf s_1,\dots \lf \mathcal{L}_{\text{int}} \lf s_n,  \ketbra{\psi_0}{\psi_0}  \otimes \rho_B \ri \ri \ri U_0^\dagger(t).
\end{equation}
Tracing out the environment degrees of freedom in Eq. \eqref{ApEq:V_tau_expansion}, we can construct the reduced density matrix as
\begin{align}
    \label{Apeq:parttrace_expansion}
    \rho_S(t) &:= \operatorname{Tr}_{\text{B}} \lf V(t,\tau) \ketbra{\psi_0}{\psi_0}  \otimes \rho_B V^\dagger(t,\tau) \ri = U_S(t) \ketbra{\psi_0}{\psi_0} U^\dagger_S(t) + \tau \rho_{\text{pert}} + O(t^{2} \tau \|H_{\text{tot}}\|^4), 
\end{align}
where the unitary $U_S$ is given by $U_S(t) = \exp \lf -i t \omega_0 \Pi_C \mathds{1}_A \otimes G \Pi_C \ri$. $\rho_{\text{pert}}$ has the explicit form,  
\begin{align}
\label{ApEq:rho_pert}
   \rho_{\text{pert}} = -\frac{1}{2}\int_0^t ds_1 \operatorname{Tr_B} \lf U_0(t) \mathcal{L}_{\text{int}}\lf s_1, \ketbra{\psi_0}{\psi_0} \otimes \rho_B\ri U_0^\dagger(t) \ri .
\end{align}
The expansion in Eq. \eqref{Apeq:parttrace_expansion} is guaranteed to converge only if each summand in Eq.~\eqref{ApEq:V_tau_expansion} has finite Schatten-one norm. This is because of the nature of upper bounds of partial traces. In a formula due to Rastegin, we are guaranteed that~\cite{rastegin_bounding_partial_trace}
\begin{equation}
    \| \operatorname{Tr}_2 (O)\|_{\infty} \leq \|O\|_1.
\end{equation}
As a consequence, we will assume that all the relevant operators have finite Schatten-1 norm. 
 Note that the term that is zeroth order in $\tau$ in Eq.~\eqref{ApEq:rho_pert} only couples the codespace with the error space, that is 
 \begin{align}
     \label{ApEq:Diag_zero}
     0 &= \Pi_C U_0(t)\left[ U_0^\dagger (s)\left[H_0, H_{\text{tot}}^{C \tc} + H_{\text{tot}}^{\tc C} \right] U_0(s), \ketbra{\psi_0}{\psi_0}  \otimes \rho_B\right] U_0^\dagger(t) \Pi_C, \\
     &= \Pi_{\tilde{C}} U_0(t)\left[ U_0^\dagger (s)\left[H_0, H_{\text{tot}}^{C \tc} + H_{\text{tot}}^{\tc C} \right] U_0(s), \ketbra{\psi_0}{\psi_0}  \otimes \rho_B\right] U_0^\dagger(t) \Pi_{\tilde{C}}  .
 \end{align}
 Using Eq.~\eqref{ApEq:Diag_zero}, we immediately see that $\Pi_C \rho_{\text{pert}} \Pi_C = \Pi_{\tilde{C}} \rho_{\text{pert}} \Pi_{\tilde{C}} = O(t^{{2}} \tau \|H_{\text{tot}}\|^{4})$. When $\tau = 0$, $\rho_S(t)$ is a pure state and has a single non-trivial eigenvalue $1$ with corresponding eigenvector $U_S(t) \ket{\psi_0}$. Using perturbation theory to calculate the eigenvalues and eigenvalues of $\rho_S(t)$ to the order of $O(\tau{\|H_{\textrm{tot}}\|})$ for $\tau \neq 0$, we first note that using $U_S(t) \ket{\psi_0} = \Pi_C U_S(t)\ket{\psi_0}$
\begin{equation}
\label{ApEq:Largest_eigval_perturb_expand}
    \lambda_0(\tau) = 1 - \tau \bra{\psi_0} U_S^\dagger(t) \rho_{\text{pert}} U_S(t) \ket{\psi_0} = 1 - O\lf t^{{2}}\tau^2 \|H_{\text{tot}}\|^{4}\ri.
\end{equation}
Thus, the largest eigenvalue changes only at the order of $\tau^2$. For the remaining eigenvalues, we can perform degenerate perturbation theory.  But there is a more elegant solution. Note that the eigenvalues can be seen as functions of $\tau$ with the properties that for $j\neq0$ 
\begin{equation}
    \lambda_{j}(0) = 0, \quad \lambda_{j}(\tau) \geq 0 ~\forall \tau \geq 0.
\end{equation}
This means that $\tau = 0$ is a local minima and hence $\partial_{\tau} \lambda_j \geq 0$. Note that we only do a one sided derivative, and hence we cannot immediately guarantee that $\lambda_j = 0$, because we only know that the eigenvalue decreases from the right and have no information on its behavior from the left. However, the additional constraint $\sum_j \lambda_j (\tau) = 1$ for all $\tau \geq 0$, guarantees that the first derivative terms are 0. Using that $\lambda_0(0) = 1$ and $\partial_\tau \lambda_0 (0) = 0$, we can show that the normalisation condition reduces to
\begin{equation}
    \tau \sum_{j \geq 1} \partial_\tau \lambda_j(0) + O(\tau^2) = 0.
\end{equation}
Since, $\partial_{\tau} \lambda_j \geq 0$, the above condition can only be satisfied if $\partial_{\tau} \lambda_j = 0$. Hence, we have shown that the eigenvalues change only at the order of $\tau^2$. Since only $\lambda_0(0) \neq 0$, to calculate the QFI to the order of $O(\tau{\|H_{\textrm{tot}}\|})$ we only need to perturbatively expand $\ket{\lambda_0}$. Noting again that $\rho_\textrm{pert}$ only couples the codespace to the error space, we have
\begin{equation}
    \ket{\lambda_0} = U_{{S}}(t) \ket{\psi_0} + \tau \sum_{k} \braket{k|\rho_{\textrm{pert}} U_S(t) |\psi_0}|k\rangle + O(\tau^2{\|H_{\textrm{tot}}\|^{2}}), 
\end{equation}
where $\ket{k}$ forms a basis on the error space, i.e, $\Pi_{\tilde{C}} = \sum_k \ketbra{k}{k}$. Consider the QFI given in the form (see Ref.~\cite{liu_QFI_review}, Corollary 2.2.1) 
\begin{equation}
    \label{ApEq:QFI}
    \mathcal{F}(t) = \sum_{ \lambda_j \neq 0} \frac{\left(\partial_{\omega_0} \lambda_j\right)^2}{\lambda_j}+\sum_{ \lambda_j \neq 0} 4 \lambda_j\left\langle\partial_{\omega_0} \lambda_j \mid \partial_{\omega_0} \lambda_j\right\rangle-\sum_{\lambda_j, \lambda_k \neq 0} \frac{8 \lambda_j \lambda_k}{\lambda_j+\lambda_k}\left|\left\langle\partial_{\omega_0} \lambda_j \mid \lambda_k\right\rangle\right|^2 .
\end{equation}
 It is easy to see that the first term is of $O(\tau^2)$ as the eigenvalues change only in that order. Similarly, it is easy to see that for all $j \neq 0$ in the second term already contribute only to the order of $O(\tau^2)$. Taking the derivative of $\ket{\lambda_0}$, we find 
\begin{equation}
    \ket{\partial_{\omega_0} \lambda_0} = -i t \Pi_C {\mathds{1}_A \otimes G}\Pi_C U_S(t) \ket{\psi_0} + \tau \sum_k \partial_{\omega_0} \braket{k|\rho_{\textrm{pert}} U_S(t) |\psi_0}|k\rangle + {O(\tau^2\|H_{\textrm{tot}}\|)}.
\end{equation}
Note that the first term is in the codespace while the second is in the error space. Thus, we show that $\braket{\partial_{\omega_0} \lambda_0|\partial_{\omega_0} \lambda_0}$ does not contribute at the order of $O(\tau)$
\begin{equation}
     \braket{\partial_{\omega_0} \lambda_0|\partial_{\omega_0} \lambda_0} = t^2 \braket{\psi_0 \Big|  \lf{U_{S}(t)^{\dagger}} \Pi_C {\mathds{1}_A \otimes G}
\Pi_C\ri^2 {U_{S}(t)}\Big| \psi_0 } + O(\tau^2).
\end{equation}
Thus, the second term also does not contribute towards $\mathcal{F}(t)$  to order $O(\tau)$. It is also easy to notice that, for the same reasons as above 
\begin{equation}
    \label{ApEq:Third_term_first_part}
    \braket{\lambda_0|\partial_{\omega_0} \lambda_0} = -i{\langle\psi_{0}|U_{S}(t)^{\dag}}\lf t \Pi_C \mathds{1}_A{\otimes G} \Pi_C \ri{U_{S}(t)|\psi_{0}\rangle} + O(\tau^2{\|H_{\textrm{tot}}\|}).
\end{equation}
Consider the third term in Eq.~\eqref{ApEq:QFI}. We can split the summation into four pieces: $j = 0, k = 0$, $j \neq 0, k \neq 0$, $j \neq 0, k = 0$ and $j \neq 0, k \neq 0$, with the latter three immediately contributing only to $O(\tau^2)$, due to the scaling of the eigenvalues. Using Eqs.~\eqref{ApEq:Third_term_first_part} and \eqref{ApEq:Largest_eigval_perturb_expand} we see that the first piece does not have terms at the $O(\tau)$. More specifically 
\begin{equation}
    \frac{4 \lambda_0^2}{\lambda_0} \left|\left\langle\partial_{\omega_0} \lambda_0 \mid \lambda_0\right\rangle\right|^2 =  4 t^2\braket{ \psi_0 \Big |{U_{S}(t)^{\dag}} \Pi_C \mathds{1}_A{\otimes G} \Pi_C{U_{S}(t)}  \Big| \psi_0 }^{{2}} + O(\tau^2).
\end{equation}
Thus, the third term also does not contribute at the order of $O(\tau)$. Combining everything, we show that 
\begin{equation}
    \mathcal{F}(t) = 4 t^2 \lf \braket{\psi_0 \Big|{U_{S}(t)^{\dag}}  \lf \Pi_C \mathds{1}_A{\otimes G} \Pi_C\ri^2 {U_{S}(t)}\Big|\psi_0 } -  \braket{ \psi_0 \Big |{U_{S}(t)^{\dag}} \Pi_C \mathds{1}_A{\otimes G} \Pi_C {U_{S}(t)} \Big| \psi_0}  \ri + O(\tau^2).
\end{equation}

\section{AQED example: spontaneous decay of a two-level system}
\label{ap:spon_decay}
This section will closely follow the derivations in Ref.~\cite{breuer_book} section 10.1. Consider the Hamiltonian under the rotating wave approximation 
\begin{equation}
    H = \omega_0 \sigma_+\sigma_- + \sum_k \omega_k b_k ^\dagger b_k + \sum_k \sigma_+ \otimes g_k b_k  + \sigma_- \otimes g^*_k b^\dagger_k. 
\end{equation}
The Hamiltonian conserves the total number of excitations as $N = \sum_k b_k^\dagger b_k + \sigma_+ \sigma_-$. Using the prescription in Lemma \ref{lem:Error_detection}, the codespace is spanned by $\ket{00} , \ket{11}$. Thus, assuming the environment begins in the vacuum, the total state of the system, ancilla and the environment can be written as 
\begin{equation}
     \label{Eq:Spont_decay_param_with_ancilla}
    \ket{\psi(t)} = \alpha(t) \ket{0}\ket{0}\ket{0}^{\otimes n} + \beta(t) \ket{1}\ket{1}\ket{0}^{\otimes n} + \sum_k d_k(t) \ket{0} \ket{1}\ket{1_k}.
\end{equation}
Following the analysis in  Ref.~\cite{breuer_book} (section 10.1), $\beta(t)$ can be represented by an integro-differential equation 
\begin{equation}
\label{ApEq:integro_diff_1}
    \dot{\beta}(t) + i  \omega_0 \beta(t) + \int_{0}^t~ d\tau  ~ \int_0^\infty J(\omega) e^{-i \omega (t - \tau)}  \beta(\tau)=0,
\end{equation}
where $J(\omega) = \sum_k |g_k|^2 \delta(\omega - \omega_k)$ is the spectral density associated with the environment. Note that Eq.~\eqref{ApEq:integro_diff_1} describes the Schrodinger picture evolution, and hence differs from the one in Ref. \cite{breuer_book}, which is in the interaction picture. We consider a Lorentzian spectral density with detuning, as used in the Jaynes-Cummings model 
\begin{equation}
    J(\omega) = \frac{1}{2 \pi} \frac{\gamma_0 \lambda^2}{(\omega_0 - \Delta - \omega)^2 + \lambda^2},
\end{equation}
where $\lambda$ denotes the spectral width of the coupling,  $\tau_B = \lambda^{-1}$ is the reservoir correlation time and $ \tau_S = \gamma^{-1}_0$ represents the characteristic system timescale. For the plots in this article, we chose $\lambda = \omega_0 $, $\gamma_0 =\Delta = 5 \omega_0 $. Solving Eq.~\eqref{ApEq:integro_diff_1} via the Laplace transform we get \cite{breuer_book} 
\begin{equation}
    \label{Eq:JC_detuning_beta}
    \beta (t) = \beta(0) e^{-(\lambda - i \Delta  + i \omega_0)t/2 } \lf \cosh \lf \frac{dt}{2} \ri + \frac{\lambda - i \Delta}{d}\sinh\lf \frac{dt}{2} \ri \ri,
\end{equation}
where $d = \sqrt{(\lambda + i \Delta)^2 - 2 \gamma_0 \lambda}$. Using Eq. \eqref{Eq:JC_detuning_beta} one can calculate the state at time $t$ for an uninterrupted evolution, and hence we can calculate the Fisher information for $\omega_0$.

\subsection*{Evolution with mid-circuit measurements and reset}
Consider measurements that are performed at intervals of $\tau \ll \lambda^{-1}$. Define $\beta_\tau := \beta(\tau)/\beta(0)$. The state just before the first measurement will be given as
\begin{equation}
     \ket{\psi(\tau)} = \alpha(0) \ket{0}\ket{0}\ket{0}^{\otimes n} + \beta(\tau) \ket{1}\ket{1}\ket{0}^{\otimes n} + \sum_k d_k(\tau) \ket{0} \ket{1}\ket{1_k}.
\end{equation}

Under the boundary conditions that $\alpha(0) = \beta(0) = 1/\sqrt{2}$, the probability of projecting into the error space is given by $\frac{1}{2} \lf 1 - |\beta_\tau|^2 \ri$. Since we reset to the initial state if the state is projected onto the error space, after $1$ measurement, the state of the system is given by the ensemble 
\begin{equation}
    \rho^{(1)} := \left\{ \frac{1}{2} \lf 1 + |\beta_\tau|^2 \ri, \ket{\psi_1}  ~; \frac{1}{2} \lf 1 - |\beta_\tau|^2 \ri, \ket{\psi_0} \right\},
\end{equation}
where 
\begin{equation}
    \ket{\psi_n} = \frac{1}{\sqrt{(1 + |\beta_\tau^n|^2)}} \lf \ket{00} + \beta_\tau^n \ket{11} \ri.
\end{equation}
Now consider the state $\ket{\psi_1}$ that has evolved for $\tau$. As before the probability that this will collapse into the error space is given by
\begin{equation}
    p_E = 1 - |\alpha(0)|^2 - |\beta(\tau)|^2 = \frac{1 + |\beta_\tau|^2 - 1 - |\beta^2_\tau|^2}{1 + |\beta_\tau|^2} = \frac{|\beta_\tau|^2 \lf 1 - |\beta_\tau|^2 \ri}{1 + |\beta_\tau|^2},
\end{equation}
where $\alpha(0) = \sqrt{(1 + |\beta_\tau|^2)}^{-1} $ and $\beta(\tau) = \beta_\tau^2 \sqrt{(1 + |\beta_\tau|^2)}^{-1}$, where we have assumed that the environment is reset whenever the state is reset. Thus, after $2$ measurements, the state of the system is given by the ensemble 
\begin{equation}
    \rho^{(2)} := \left\{ \frac{1}{2} \lf 1 + |\beta_\tau|^4 \ri, \ket{\psi_2}  ~; \frac{1}{4} \lf1 - |\beta_\tau|^4 \ri, \ket{\psi_1} ~; \frac{1}{4} \lf1 - |\beta_\tau|^4 \ri, \ket{\psi_0} \right\}.
\end{equation}

Continuing the pattern, one can prove using induction that after $n$ rounds of measurement the state of the system is given by, 
\begin{equation}
    \rho^{(n)} = \{p_n, \ket{\psi_n} ~; p_{n-1} \ket{\psi_{n - 1}} ~; \dots ~;p_0 \ket{\psi_0}\},
\end{equation}
where the probabilities $p_{j}$ are given by, 

\begin{equation}
    p_j  = \left \{ \begin{matrix}
         \frac{1}{2} \lf 1 + |\beta_\tau|^{2n} \ri  & \text{if }~~ j = n, \\
         \frac{1}{4} \lf 1 -  |\beta_\tau|^{2} \ri \lf 1+ |\beta_\tau|^{2j} \ri \left[\frac{1}{2} \lf 1 + |\beta_\tau|^{2} \ri\right]^{n - j - 1} & \text{otherwise}.
    \end{matrix} \right.
\end{equation}

Thus, after $n$ rounds of measurements we have the ensemble decomposition of the density matrix. We use this to calculate the average QFI $F_{\mathrm{avg}}(t)$, as 
\begin{equation}
   F_{\mathrm{avg}}(t) = \sum_{i = 1}^n p_i \mathcal{F} \lf \ketbra{\psi_i}{\psi_i}\ri.  
\end{equation}
If we keep track of the measurement record, we know the exact state $\ket{\psi_n}$ after the entire evolution. In this case, the average QFI is the expected QFI over all possible measurement records (or trajectories). If we forget the measurement record, we only have access to the density matrix $\rho^{(n)}$. Using the convexity of the QFI, we can show that for this case the average QFI upper bounds the QFI. 

\section{Simulations of AQEC and AQED for perpendicular dephasing}
\label{ap:Perp_dephasing_sim}
 \begin{figure}[htbp!]
     \centering
     \includegraphics[width=0.48\linewidth]{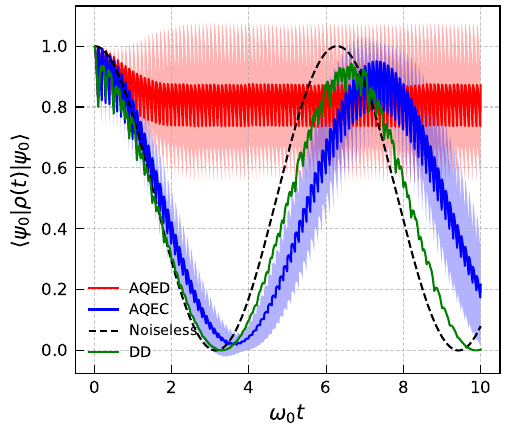}
     \hfill
     \includegraphics[width=0.48\linewidth]{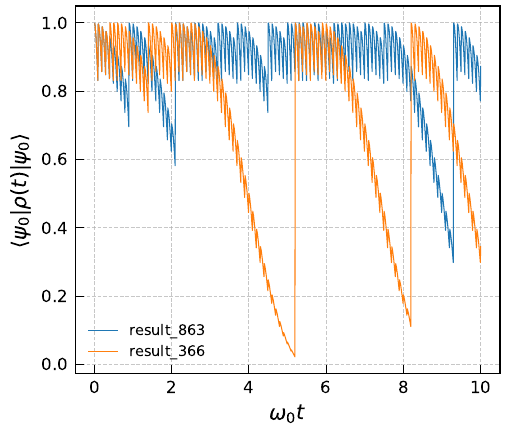}
     \caption{(Left) We plot the mean evolution of $10000$ different trajectories along with the standard deviation for $\omega_0 \delta t = 0.1$. The AQEC and DD evolutions are slightly damped, compared to the noiseless evolution, while the AQED evolution is extremely noisy. (Right) We plot two randomly chosen trajectories for the AQED simulation.}
     \label{Apfig:Perp_dephase_0p1}
 \end{figure}
In this section, we present the numerical simulation of the perpendicular dephasing model used to support our analysis regarding the robustness of the AQEC scheme in comparison with the AQED scheme. For computational tractability, we model the environmental modes as qubits connected to the system in a star topology. The system and ancilla are also represented as qubits. For the plots presented in this article, the number of environmental modes is fixed to $5$. The total Hamiltonian is given by 
\begin{equation}
     H = \omega_0 \mathds{1}_A \otimes \sigma_z \otimes \mathds{1}_E + \sum_j g_{env} (j) \mathds{1}_A \otimes \sigma_x \otimes \sigma_x^{(j)} + \sum_j \Omega_{env}(j) \mathds{1}_A \otimes \mathds{1}_{S} \otimes \sigma_z^{(j)},
\end{equation}
where $\sigma_x^{(i)}$ acts on the $i^{th}$ environmental mode, and the summation runs over all modes. For the plots in this paper, we fix the mode frequencies and the coupling constants as 
\begin{equation}
      \vec{\Omega}_{env} = [-0.9881\omega_{0},1.7832\omega_{0},-12.6627\omega_{0},-5.2682\omega_{0},-0.3333\omega_{0}], \quad  \vec{g}_{env} = (g/\sqrt{5})[\omega_{0},\omega_{0},\omega_{0},\omega_{0},\omega_{0}].
\end{equation}
The mode frequencies were randomly sampled from a standard Cauchy distribution with location parameter $x_0 = 0$ and scale parameter $s = 1$, to avoid any bias due to regularity. The parameters were chosen such that, in the asymptotic limit of infinite environmental modes, the spectral density approaches a Lorentzian distribution multiplied by a factor of $g$. The signal and the error generators are given by 
 \begin{equation}
    G = \sigma_z, \quad E_1 = \sigma_x,
\end{equation}
 which satisfy the conditions to construct an AQEC code. Since the error generator is $\sigma_x$, the above Hamiltonian generates dephasing along the \textit{x}-axis. The initial state is given by $\ket{\psi_0}$
\begin{equation}
    \ket{\psi_0} = \frac{1}{\sqrt{2}} \lf \ket{00} + \ket{11}\ri \otimes ]\ket{0}^{\otimes 5}.
\end{equation}
 \begin{figure}[htbp!]
     \centering
     \includegraphics[width=0.48\linewidth]{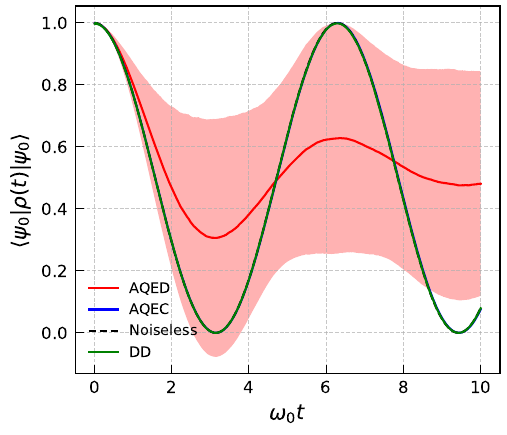}
     \hfill
     \includegraphics[width=0.48\linewidth]{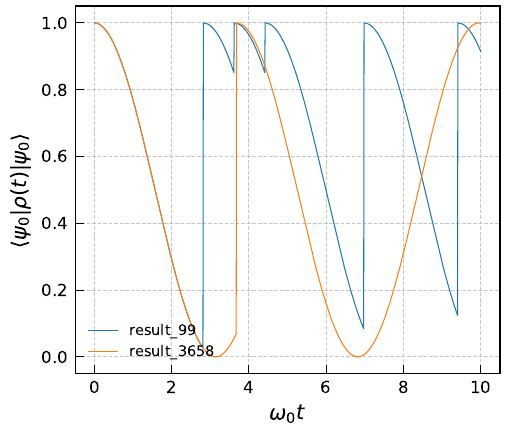}
     \caption{(Left) We plot the mean evolution of a $10000$ different trajectories along with its standard deviation for $\omega_0 \delta t = 0.01$. The AQEC and DD evolutions closely match the noiseless evolution, while the AQED evolution is still damped. (Right) We plot two trajectories for the AQED simulation. The big jumps back to a fidelity of $1$ corresponds to the time steps, where the system was projected into the error subspace and reset back to the initial state.}
     \label{Apfig:Perp_dephase_0p01}
 \end{figure} 
 We then perform three distinct evolutions: an AQED-type simulation, where the state is reset each time it is projected into the error space, an AQEC-type simulation, where the recovery operator (in this case, $\sigma_x$ on the system) is applied, and a DD-type evolution where a unitary control is implemented instead of the measurement. The AQED and AQEC evolutions are simulated  using a stochastic trajectory approach. At each measurement step, a random number $q$ is drawn uniformly from the interval $[0,1]$. If $q < p_{\text{err}}$, the probability of projecting into the error subspace, the state is projected into the error subspace. Otherwise, we project it into the codespace. The plots shown here are for a $10000$ different trajectory runs. We evaluate the fidelity between the reduced system-ancilla state and the initial Bell state, which effectively captures the Rabi oscillations of the system.  In Figs.~\ref{Apfig:Perp_dephase_0p1},~\ref{Apfig:Perp_dephase_0p01}, and~\ref{Apfig:Perp_dephase_0p001}, we plot the fidelity of state at time $t$ with the initial state against the total time $t$ for the measurement intervals, $\omega_0\delta t = 0.1$, $0.01$, and $0.001$ respectively, for $g = 10\omega_{0}$. This captures the Rabi oscillation of the state whose frequency carries information about the parameter. The solid lines represent the means and the shaded regions represent the standard deviations for $10000$ trajectories.  We also plot a few individual trajectories for the AQED evolution to highlight the effect of reset recovery.  \par
  \begin{figure}[htbp!]
     \centering
     \includegraphics[width=0.48\linewidth]{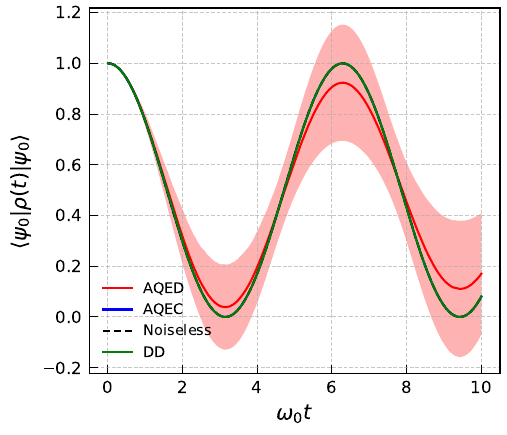}
     \hfill
     \includegraphics[width=0.48\linewidth]{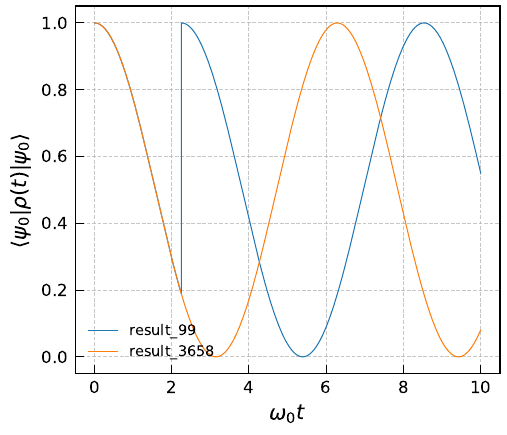}
     \caption{(Left) We plot the mean evolution of a $10000$ different trajectories along with its standard deviation for $\omega_0 \delta t = 0.001$. The AQEC and DD evolutions very closely matche the noiseless evolution, while the AQED evolution is still slightly damped. (Right) We plot two randomly chosen trajectories for the AQED simulation. The big jumps back to a fidelity of $1$ corresponds to the time steps where the system was projected into the error subspace and reset back to the initial state.}
     \label{Apfig:Perp_dephase_0p001}
 \end{figure}
 From the Figs.~\ref{Apfig:Perp_dephase_0p1},~\ref{Apfig:Perp_dephase_0p01}, and~\ref{Apfig:Perp_dephase_0p001}, we clearly see that the AQEC and the DD schemes out perform the AQED schemes for all the three different measurement intervals. The big jumps in AQED trajectories back to a fidelity of $1$ corresponds to the time steps, where the system was projected into the error subspace and reset back to the initial state. Note that the parameters considered here place the simulation beyond the regime of leading order calculations, and hence, our claims regarding the robustness of schemes also seem to hold true for a larger range of parameters.
 \begin{figure}
     \centering
     \includegraphics[width=0.48\linewidth]{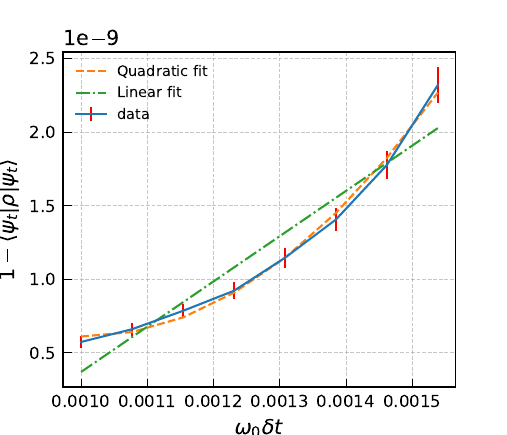}
     \includegraphics[width=0.48\linewidth]{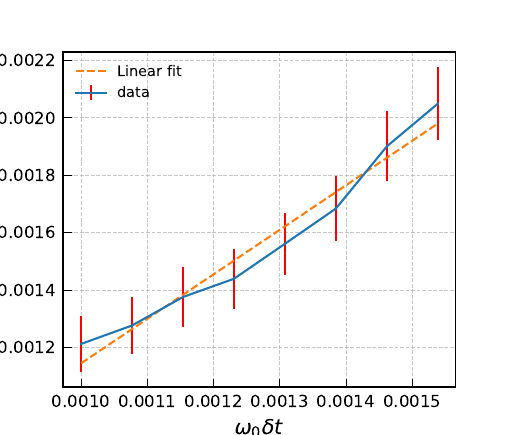}
     \includegraphics[width=0.48\linewidth]{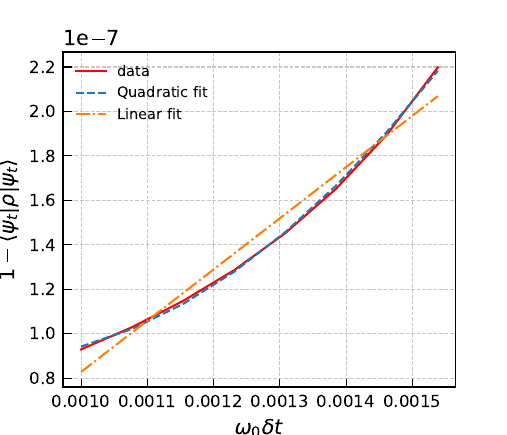}
     \caption{Infidelity of the state at time $\omega_0 t = 10$ ($\rho(t)$) with the exact evolution ($\ket{\psi_t}$) vs $\omega_{0}\delta t$. The upper-left plot shows the AQEC evolution, while the upper right plots the AQED evolution. The red bars are standard error $\sigma/\sqrt{N}$ for $N = 10000$ trajectories. The bottom plot shows the DD evolution. We plot linear and quadratic fits generated using the \textit{polyfit} function in \textit{numpy} library of python.}
     \label{Apfig:scaling_plots}
 \end{figure}
Finally, we plot the infidelity of the state at time $\omega_{0}t = 10$ with the state resulting from noiseless evolution $\ket{\psi_t}$ against different measurement time steps $\delta t$ for $g = \omega_{0}$ in Fig.~\ref{Apfig:scaling_plots}. We clearly see that the infidelity scales as $O(\delta t^2)$ for AQEC and the DD evolutions, while it scales as $O(\delta t)$ for AQED evolution.

\end{document}